%% file: main.tex
\documentclass[11pt,letterpaper]{article}

\usepackage{typearea}
\typearea{15}

\usepackage{theorem,latexsym,graphicx}
\usepackage{amsmath,amssymb,enumerate}
\usepackage{xspace}
\usepackage{bm}
\usepackage{ifpdf}
\usepackage{shadow,shadethm,color}
\usepackage[compact]{titlesec}
\usepackage{algorithm}
\usepackage[noend]{algorithmic}
\usepackage{subfigure}
\usepackage{verbatim}
\usepackage{times}
\usepackage{paralist}

\definecolor{Darkblue}{rgb}{0,0,0.4}
\definecolor{Brown}{cmyk}{0,0.81,1.,0.60}
\definecolor{Purple}{cmyk}{0.45,0.86,0,0}
\newcommand{\mydriver}{hypertex}
\ifpdf
 \renewcommand{\mydriver}{pdftex}
\fi
\usepackage[breaklinks,\mydriver]{hyperref}
\hypersetup{colorlinks=true,
            citebordercolor={.6 .6 .6},linkbordercolor={.6 .6 .6},%
citecolor=Darkblue,urlcolor=black,linkcolor=Darkblue,pagecolor=black}
\newcommand{\lref}[2][]{\hyperref[#2]{#1~\ref*{#2}}}


\newtheorem{theorem}{Theorem}[section]

\newtheorem{lemma}[theorem]{Lemma}
\newshadetheorem{lemmashaded}[theorem]{Lemma}

\newtheorem{observation}[theorem]{Observation}
\newtheorem{claim}[theorem]{Claim}

\numberwithin{algorithm}{section}

\newenvironment{proof}{

\noindent{\bf Proof:}}
{\hfill$\blacksquare$

}

\newenvironment{proofof}[1]{

\noindent{\bf Proof of {#1}:}}
{\hfill$\blacksquare$

}

\newcommand{\junk}[1]{}
\newcommand{\ignore}[1]{}

\newcommand{\poly}{\operatorname{poly}}
\newcommand{\argmin}{\operatorname{argmin}}

\newcommand{\sse}{\subseteq}

\newcommand{\Si}{{\mathcal{S}_i}}
\renewcommand{\S}{{\mathcal{S}}}

\newcommand{\ts}{\textstyle}

\newcommand{\Opt}{\ensuremath{\mathsf{Opt}\xspace}}
\newcommand{\LPOpt}{\ensuremath{\mathsf{LPOpt}\xspace}}

\newcommand{\mab}{\ensuremath{\mathsf{MAB}}\xspace}
\newcommand{\sks}{\ensuremath{\mathsf{StocK}}\xspace}

\newcommand{\parent}{\mathsf{parent}}

\newcommand{\head}{\mathsf{Head}}
\newcommand{\comp}{\mathsf{comp}}
\newcommand{\ptime}{\mathsf{time}}

\newcommand{\prob}{\mathsf{prob}}
\newcommand{\arm}{\mathsf{arm}}
\newcommand{\lpt}{\mathbb{T}}
\newcommand{\depth}{\mathsf{depth}}
\newcommand{\factor}{3}

\newcommand{\E}{\mathbb{E}}

\newcounter{note}[section]

\newcommand{\qedsymb}{\hfill{\rule{2mm}{2mm}}}
\renewenvironment{proof}{\begin{trivlist} \item[\hspace{\labelsep}{\bf
\noindent Proof.\/}] }{\qedsymb\end{trivlist}}%

\newcommand{\initOneLiners}{%
    \setlength{\itemsep}{0pt}
    \setlength{\parsep }{0pt}
    \setlength{\topsep }{0pt}
}
\newenvironment{OneLiners}[1][\ensuremath{\bullet}]
    {\begin{list}
        {#1}
        {\initOneLiners}}
    {\end{list}}

\begin{document}

\title{Approximation Algorithms for Correlated Knapsacks \\ and
  Non-Martingale Bandits}
\author{
Anupam Gupta\thanks{Deparment of Computer Science, Carnegie Mellon University, Pittsburgh
    PA 15213.}
\and
Ravishankar Krishnaswamy$^*$
\and
Marco Molinaro\thanks{Tepper School of Business, Carnegie Mellon University, Pittsburgh
    PA 15213.}
\and
R. Ravi$^\dagger$}
\date{}

\maketitle

\begin{abstract}

In the stochastic knapsack problem, we are given a knapsack of size $B$,
and a set of jobs whose sizes and rewards are drawn from a known
probability distribution. However, the only way to know the actual size
and reward is to schedule the job---when it completes, we get to know
these values. How should we schedule jobs to maximize the expected total
reward? We know constant-factor approximations for this problem when we
assume that rewards and sizes are independent random variables, and that
we cannot prematurely cancel jobs after we schedule them. What can we
say when either or both of these assumptions are changed?

The stochastic knapsack problem is of interest in its own right, but
techniques developed for it are applicable to other stochastic packing
problems. Indeed, ideas for this problem have been useful for budgeted
learning problems, where one is given several arms which evolve in a
specified stochastic fashion with each pull, and the goal is to pull the
arms a total of $B$ times to maximize the reward obtained.  Much recent
work on this problem focus on the case when the evolution of the arms
follows a martingale, i.e., when the expected reward from the future is
the same as the reward at the current state.  What can we say when the
rewards do not form a martingale?

In this paper, we give constant-factor approximation algorithms for the
stochastic knapsack problem with correlations and/or cancellations, and
also for budgeted learning problems where the martingale condition
is not satisfied, using similar ideas. Indeed, we can show that
previously proposed linear programming relaxations for these problems
have large integrality gaps. We propose new time-indexed LP relaxations;
using a decomposition and ``gap-filling'' approach, we convert these
fractional solutions to distributions over strategies, and then use the
LP values and the time ordering information from these strategies to
devise a randomized adaptive scheduling algorithm.  We hope our LP formulation
and decomposition methods may provide a new way to address other
correlated bandit problems with more general contexts.

\end{abstract}

\thispagestyle{empty}
\setcounter{page}{0}
\newpage

\input{blurb}

\input{knap}

\section{Multi-Armed Bandits}
\label{sec:mab}

We now turn our attention to the more general Multi-Armed Bandits
problem (\mab). In this framework, there are $n$ \emph{arms}: arm $i$
has a collection of states denoted by $\Si$, a starting state $\rho_i
\in \Si$; Without loss of generality, we assume that $\Si \cap \mathcal{S}_j =
\emptyset$ for $i \neq j$. Each arm also has a \emph{transition graph}
$T_i$, which is given as a polynomial-size (weighted) directed tree
rooted at $\rho_i$; we will relax the tree assumption later. If
there is an edge $u \to v$ in $T_i$, then the edge weight $p_{u,v}$
denotes the probability of making a transition from $u$ to $v$ if we
play arm $i$ when its current state is node $u$; hence $\sum_{v: (u,v)
  \in T_i} p_{u,v} =1$. Each time we play an arm, we get a reward whose
value depends on the state from which the arm is played. Let us denote
the reward at a state $u$ by $r_u$. Recall that the martingale property on rewards requires that $\sum_{v: (u,v) \in T_i} p_{u,v} r_v = r_u$ for all states $u$.

{\bf Problem Definition.} For a concrete example, we consider the
following budgeted learning problem on \emph{tree transition graphs}. Each of the arms starts at the start state $\rho_i
\in \Si$. We get a reward from each of the states we play, and the goal
is to maximize the total expected reward, while not exceeding a
pre-specified allowed number of plays $B$ across all arms. The framework
described below can handle other problems (like the explore/exploit
kind) as well, and we discuss this in \lref[Appendix]{xsec:mab}.

Note that the Stochastic Knapsack problem considered in the previous
section is a special case of this problem where each item corresponds to
an arm, where the evolution of the states corresponds to the explored
size for the item. Rewards are associated with each stopping size, which
can be modeled by end states that can be reached from the states of the
corresponding size with the probability of this transition being the
probability of the item taking this size. Thus the resulting trees are
paths of length up to the maximum size $B$ with transitions to end
states with reward for each item size.
For example, the transition graph in \lref[Figure]{fig:redn} corresponds to an item which instantiates to a size of $1$ with probability $1/2$ (and fetches a reward $R_1$), takes size $3$ with probability $1/4$ (with reward $R_3$), and size $4$ with the remaining probability $1/4$ (reward is $R_4$). Notice that the reward on stopping at all intermediate nodes is $0$ and such an instance therefore does not satisfy the martingale property. Even though the rewards are obtained in this example on reaching a state rather than playing it, it is not hard to modify our methods for this version as well.

\begin{figure}[ht]
\centering
\includegraphics[scale=0.4]{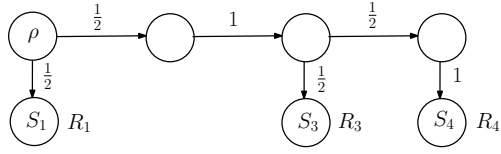}
\caption{Reducing Stochastic Knapsack to MAB}
\label{fig:redn}
\end{figure}

\paragraph{Notation.}
The transition graph $T_i$ for arm $i$ is an out-arborescence defined on
the states $\Si$ rooted at $\rho_i$.  Let $\depth(u)$ of a node $u \in
\Si$ be the depth of node $u$ in tree $T_i$, where the root $\rho_i$ has
depth $0$. The unique parent of node $u$ in $T_i$ is denoted by
$\parent(u)$.  Let $\S = \cup_{i} \Si$ denote the set of all states in
the instance, and $\arm(u)$ denote the arm to which state $u$ belongs,
i.e., the index $i$ such that $u \in \Si$. Finally, for $u \in \Si$, we
refer to the act of playing arm $i$ when it is in state $u$ as ``playing
state $u \in \Si$'', or ``playing state $u$'' if the arm is clear in
context.

\subsection{Global Time-indexed LP}
\label{sec:mab-lp}


In the following, the variable $z_{u,t} \in [0,1]$ indicates that the
algorithm plays state $u \in \Si$ at time $t$.  For state $u \in \Si$
and time $t$, $w_{u,t} \in [0,1]$ indicates that arm $i$ \emph{first
  enters} state $u$ at time $t$: this happens if and only if the
algorithm \emph{played} $\parent(u)$ at time $t-1$ and the arm made a
transition into state $u$.

\begin{alignat}{2} \tag{$\mathsf{LP}_\mathsf{mab}$} \label{lp:mab}
	\max \ts \sum_{u,t} r_u &\cdot z_{u,t}\\
	w_{u,t} &= z_{\parent(u), t-1}  \cdot p_{\parent(u),u} & \qquad \forall t \in [2,B],\, u \in \S \setminus \cup_{i} \{\rho_i\} \label{eq:mablp1}\\
	\ts \sum_{t' \le t} w_{u,t'} &\geq \ts \sum_{t' \leq t} z_{u,t'} & \qquad \forall t \in [1,B], \, u \in \S \label{eq:mablp2}\\
	\ts \sum_{u \in \S} z_{u,t} &\le 1 & \qquad \forall t \in [1,B]  \label{eq:mablp3}\\
	w_{\rho_i, 1} &= 1 & \qquad \forall i \in [1,n] \label{eq:mablp4}
\end{alignat}

\begin{lemma}
The value of an optimal LP solution $\LPOpt$ is at least $\Opt$, the expected reward of an optimal adaptive strategy.
\end{lemma}

\begin{proof}
 We convention that $\Opt$ starts playing at time $1$. Let $z^*_{u,t}$ denote the probability that $\Opt$ plays state $u$ at time $t$, namely, the probability that arm $\arm(u)$ is in state $u$ at time $t$ and is played at time $t$. Also let $w^*_{u,t}$ denote the probability that $\Opt$ ``enters'' state $u$ at time $t$, and further let $w^*_{\rho_i,1} = 1$ for all $i$.

	We first show that $\{z^*, w^*\}$ is a feasible solution for \ref{lp:mab} and later argue that its LP objective is at least $\Opt$. Consider constraint \eqref{eq:mablp1} for some $t \in [2, B]$ and $u \in \S$. The probability of entering state $u$ at time $t$ conditioned on $\Opt$ playing state $\parent(u)$ at time $t - 1$ is $p_{\parent(u),u}$. In addition, the probability of entering state $u$ at time $t$ conditioning on $\Opt$ not playing state $\parent(u)$ at time $t - 1$ is zero. Since $z^*_{\parent(u),t-1}$ is the probability that $\Opt$ plays state $\parent(u)$ at time $t - 1$, we remove the conditioning to obtain $w^*_{u,t} = z^*_{\parent(u),t-1} \cdot p_{\parent(u),u}$.
	
	Now consider constraint \eqref{eq:mablp2} for some $t \in [1, B]$ and $u \in \S$. For any outcome of the algorithm (denoted by a sample path $\sigma$),
 let $\mathbf{1}^{enter}_{u',t'}$ be the indicator variable that $\Opt$ enters state $u'$ at time $t'$ and let $\mathbf{1}^{play}_{u',t'}$ be the indicator variable that $\Opt$ plays state $u'$ at time $t'$. Since $T_i$ is acyclic, state $u$ is played at most once in $\sigma$ and is also entered at most once in $\sigma$. Moreover, whenever $u$ is played before or at time $t$, it must be that $u$ was also entered before or at time $t$, and hence $\sum_{t' \le t} \mathbf{1}^{play}_{u,t'} \le \sum_{t' \le t} \mathbf{1}^{enter}_{u, t'}$. Taking expectation on both sides and using the fact that $\E[\mathbf{1}^{play}_{u,t'}] = z^*_{u,t'}$ and $\E[\mathbf{1}^{enter}_{u,t'}] = w^*_{u,t'}$, linearity of expectation gives $\sum_{t' \le t} z^*_{u,t'} \le \sum_{t' \le t} w^*_{u,t'}$.
	
	To see that constraints \eqref{eq:mablp3} are satisfied, notice that we can play at most one arm (or alternatively one state) in each time step, hence $\sum_{u \in \S} \mathbf{1}^{play}_{u,t} \le 1$ holds for all $t \in [1, B]$; the claim then follows by taking expectation on both sides as in the previous paragraph.   Finally, constraints \eqref{eq:mablp4} is satisfied by definition of the start states.
	
	To conclude the proof of the lemma, it suffices to show that $\Opt = \sum_{u,t} r_u \cdot z^*_{u,t}$. Since $\Opt$ obtains reward $r_u$ whenever it plays state $u$, it follows that $\Opt$'s reward is given by $\sum_{u,t} r_u \cdot \mathbf{1}^{play}_{u,t}$; by taking expectation we get $\sum_{u,t} r_u z^*_{u,t} = \Opt$, and hence $\LPOpt \geq \Opt$.
\end{proof}


\subsection{The Rounding Algorithm}

In order to best understand the motivation behind our rounding algorithm, it would be useful to go over the example which illustrates the necessity of preemption (repeatedly switching back and forth between the different arms) in \lref[Appendix]{sec:preemption-gap}.

At a high level, the rounding algorithm proceeds as follows. In Phase~I,
given an optimal LP solution, we decompose the fractional solution for each arm
into a convex\footnote{Strictly speaking, we do not get convex
  combinations that sum to one; our combinations sum to $\sum_t
  z_{\rho_i, t}$, the value the LP assigned to pick to play the root of
  the arm over all possible start times, which is at most one.}
combination of integral ``strategy forests'' (which are depicted in
\lref[Figure]{fig:treeforest}): each of these tells us at what times to
play the arm, and in which states to abandon the arm. Now, if we sample
a random strategy forest for each arm from this distribution, we may end
up scheduling multiple arms to play at some of the timesteps, and hence
we need to resolve these conflicts. A natural first approach might be to (i) sample a strategy forest for each arm, (ii) play these arms in a random order, and (iii) for any arm follow the decisions (about whether to abort or continue playing) as suggested by the sampled strategy forest. In essence, we are ignoring the times at which the sampled strategy forest has
scheduled the plays of this arm and instead playing this arm continually
until the sampled forest abandons it. While such a non-preemptive strategy works when the martingale property holds, the example in \lref[Appendix]{sec:preemption-gap} shows that preemption is unavoidable.

Another approach would be to try to play the sampled
forests at their prescribed times; if multiple forests want to play at
the same time slot, we round-robin over them. The expected number of
plays in each timestep is 1, and the hope is that round-robin will
not hurt us much. However, if some arm needs $B$ contiguous steps to get to a
state with high reward, and a single play of some other arm gets
scheduled by bad luck in some timestep, we would end up getting nothing!

Guided by these bad examples, we try to use the continuity information
in the sampled strategy forests---once we start playing some contiguous
component (where the strategy forest plays the arm in every consecutive
time step), we play it to the end of the component. The na\"{\i}ve
implementation does not work, so we first alter the LP solution to get
convex combinations of ``nice'' forests---loosely, these are forests
where the strategy forest plays contiguously in almost all timesteps, or
in at least half the timesteps. This alteration is done in Phase~II, and
then the actual rounding in Phase~III, and the analysis appears in
\lref[Section]{sec:phase-iii}.

\subsubsection{Phase I: Convex Decomposition}
\label{sec:phase-i}

In this step, we decompose the fractional solution into a convex
combination of ``forest-like strategies'' $\{\lpt(i,j)\}_{i,j}$,
corresponding to the $j^{th}$ forest for arm $i$. We first formally
define what these forests look like:
The $j^{th}$ \emph{strategy forest} $\lpt(i,j)$ for arm $i$ is an
assignment of values $\ptime(i,j,u)$ and $\prob(i,j,u)$ to each state $u
\in \Si$ such that:
\begin{OneLiners}
\item[(i)] For $u \in \Si$ and $v = \parent(u)$, it holds that
  $\ptime(i,j,u) \geq 1+ \ptime(i,j,v)$, and
\item[(ii)] For $u \in \Si$ and $v = \parent(u)$, if $\ptime(i,j,u) \neq
  \infty$ then $\prob(i,j,u) = p_{v,u}\,\prob(i,j,v)$; else if
  $\ptime(i,j,u) = \infty$ then $\prob(i,j,u) = 0$.
\end{OneLiners}
We call a triple $(i,j,u)$ a \emph{tree-node} of $\lpt(i,j)$. When $i$ and $j$ are understood from the context, we identify the tree-node $(i,j,u)$ with the state $u$.

For any state $u$, the values $\ptime(i,j,u)$ and $\prob(i,j,u)$ denote
the time at which the arm $i$ is played at state $u$, and the
probability with which the arm is played, according to the strategy
forest $\lpt(i,j)$.\footnote{When $i$ and $j$ are clear from the
  context, we will just refer to state $u$ instead of the triple $(i,j,u)$.}  The probability values are particularly simple: if
$\ptime(i,j,u) = \infty$ then this strategy does not play the arm at
$u$, and hence the probability is zero, else $\prob(i,j,u)$ is equal to
the probability of reaching $u$ over the random transitions according to
$T_i$ if we play the root with probability $\prob(i,j,\rho_i)$. Hence, we can compute $\prob(i,j,u)
$ just given $\prob(i,j,
\rho_i)$ and whether or not $\ptime(i,j,u) = \infty$. Note that the
$\ptime$ values are not necessarily consecutive, plotting these on the timeline and connecting a state to its parents only when they are in consecutive timesteps (as
in \lref[Figure]{fig:treeforest}) gives us forests, hence the name.

\begin{figure}[ht]
\centering
\subfigure[Strategy forest: numbers are $\ptime$s]{
\includegraphics[scale=0.5]{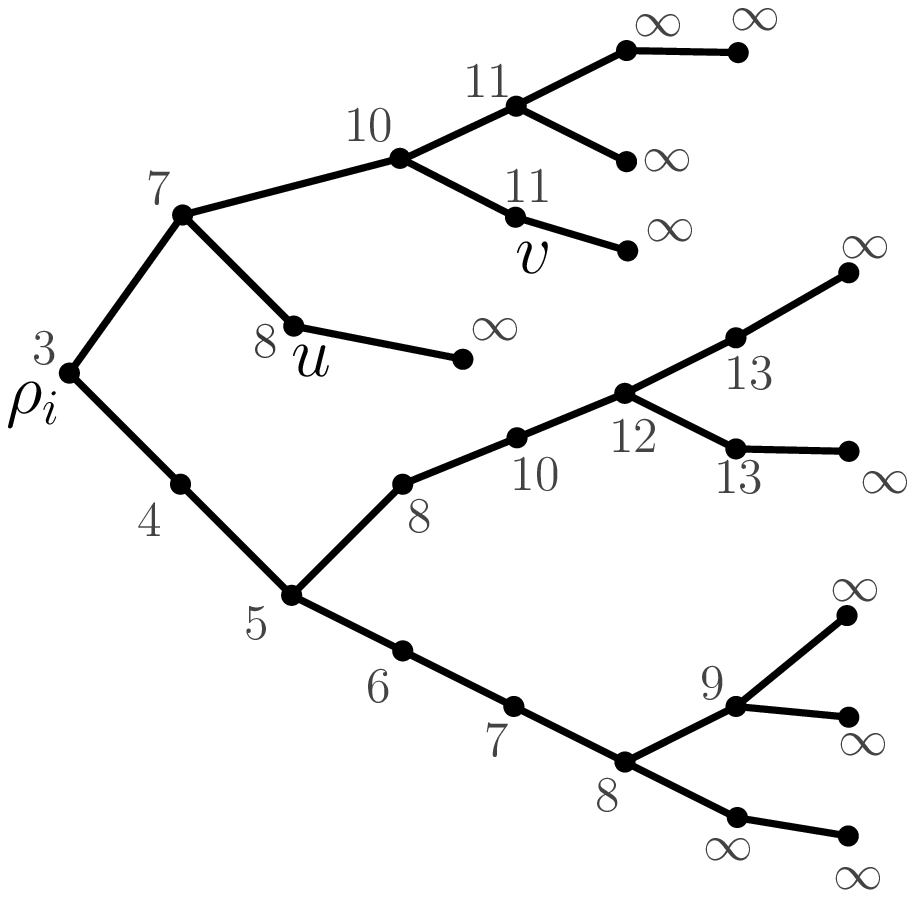}
\label{fig:subfig1}
}
\hspace{20pt}
\subfigure[Strategy forest shown on a timeline]{
\includegraphics[scale=0.5]{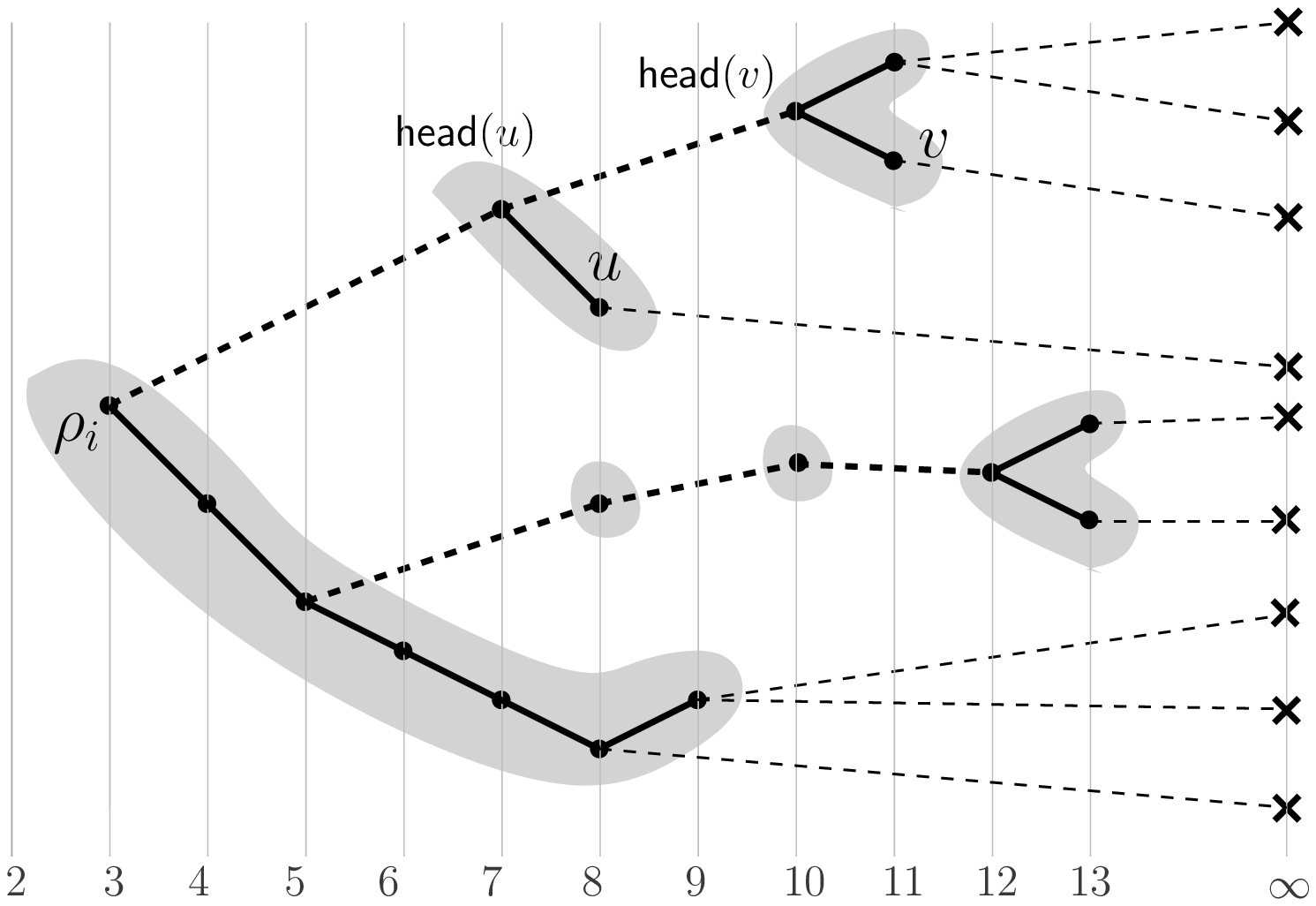}
\label{fig:subfig2}
}
\caption{Strategy forests and how
  to visualize them: grey blobs are connected components.}
\label{fig:treeforest}
\end{figure}


The algorithm to construct such a decomposition proceeds in rounds for
each arm $i$; in a particular round, it ``peels'' off such a strategy as
described above, and ensures that the residual fractional solution
continues to satisfy the LP constraints, guaranteeing that we can repeat
this process, which is similar to (but slightly more involved than)
performing flow-decompositions. The decomposition lemma is proved in
\lref[Appendix]{sec:details-phase-i}:
\begin{lemma}
  \label{lem:convexppt}
  Given a solution to~(\ref{lp:mab}), there exists a collection of
  at most $nB|\S|$ strategy forests $\{\lpt(i,j)\}$ such that $z_{u,t} =
  \sum_{j:\ptime(i,j,u) = t} \prob(i,j,u)$.\footnote{To reiterate, even though we call this a convex decomposition, the sum of the probability values of the root state of any arm is at most one by constraint~\ref{eq:mablp3}, and hence the sum of the probabilities of the root over the decomposition could be less than one in general.} Hence, $\sum_{(i, j, u):
    \ptime(i,j,u)=t} \prob(i,j,u) \leq 1$ for all $t$.
\end{lemma}

For any $\lpt(i,j)$, these $\prob$ values satisfy a ``preflow''
condition: the in-flow at any node $v$ is always at least the out-flow, namely $\prob(i,j,v) \ge \sum_{u: \parent(u)=v} \prob(i,j,u)$. This leads to
the following simple but crucial observation.

\begin{observation}
  \label{obs:treeflow}
  For any arm $i$, for any set of states $X \sse \Si$ such that no state
  in $X$ is an ancestor of another state in $X$ in the transition tree $T_i$, and
  for any $z \in \Si$ that is an ancestor of all states in $X$,
  $\prob(i,j,z) \geq \sum_{x \in X} \prob(i,j,x)$.

  More generally, given similar conditions on $X$, if $Z$ is a set of
  states such that for any $x \in X$, there exists $z \in Z$ such that
  $z$ is an ancestor of $x$, we have $\sum_{z \in Z} \prob(i,j,z) \geq
  \sum_{x \in X} \prob(i,j,x)$
\end{observation}

\subsubsection{Phase II: Eliminating Small Gaps}
\label{sec:phase-ii}

While \lref[Appendix]{sec:preemption-gap} shows that preemption is necessary
to remain competitive with respect to $\Opt$, we also should not get
``tricked'' into switching arms during very short breaks taken by the
LP. For example, say, an arm of length $(B-1)$ was played in two
continuous segments with a gap in the middle. In this case, we should
not lose out on profit from this arm by starting some other arms' plays
during the break. To handle this issue, whenever some path on the strategy
tree is almost contiguous---i.e., gaps on it are relatively small---we
make these portions completely contiguous. Note that we will not make the
entire tree contiguous, but just combine some sections together.

Before we make this formal, here is some useful notation:
Given $u \in \Si$, let $\head(i,j,u)$ be its ancestor node $v \in
\Si$ of least depth such that the plays from $v$ through $u$ occur in consecutive $\ptime$ values. More formally, the path $v = v_1, v_2,
\ldots, v_l = u$ in $T_i$ is such that $\ptime(i,j,v_{l'}) = \ptime(i,j,v_{l' - 1}) + 1$
for all $l' \in [2, l]$. We also define the \emph{connected component}
of a node $u$, denoted by $\comp(i,j,u)$, as the set of all nodes $u'$ such
that $\head(i,j,u) = \head(i,j,u')$. \lref[Figure]{fig:treeforest} shows the
connected components and heads.

The main idea of our \emph{gap-filling} procedure is the following: if a head state $v = \head(i,j,u)$ is played at time $t = \ptime(i,j,v)$ s.t. $t < 2 \cdot \depth(v)$, then we ``advance'' the $\comp(i,j,v)$ and get rid of the gap between $v$ and its parent (and recursively apply this rule)\footnote{The intuition is that such vertices have only a small gap in their play and should rather be played contiguously.}. The procedure can be described in more detail as follows.

\begin{algorithm}[ht!]
  \caption{Gap Filling Algorithm \textsf{GapFill}}
  \begin{algorithmic}[1]
  \label{alg:ridgaps}
    \FOR{$\tau$ $=$ $B$ to $1$}
    \WHILE{there exists a tree-node $u \in \lpt(i,j)$ such that $\tau = \ptime(\head(u)) <  2 \cdot \depth(\head(u))$} \label{alg:gap1}
      \STATE {\bf let} $v = \head(u)$. \label{alg:setV}
      \IF{$v$ is not the root of $\lpt(i,j)$}
    	  \STATE {\bf let} $v' = \parent(v)$.
        \STATE {\bf advance} the component $\comp(v)$ rooted at $v$ such that $\ptime(v) \leftarrow \ptime(v') + 1$, to make $\comp(v)$ contiguous with the ancestor forming one larger component. Also alter the $\ptime$s of $w \in \comp(v)$ appropriately to maintain contiguity with $v$ (and now with $v'$).
      \ENDIF
    \ENDWHILE \label{alg:gap3}
    \ENDFOR
\end{algorithmic}
\end{algorithm}

One crucial property is that these ``advances'' do not increase by much the number of plays that occur at any given time $t$. Essentially this is because if for some time slot $t$ we ``advance'' a set of components that were originally scheduled after $t$ to now cross time slot $t$, these components moved because their ancestor paths (fractionally) used up at least $t/2$ of
the time slots before $t$; since there are $t$ time slots to be used up,
each to unit extent, there can be at most $2$ units of components being
moved up. Hence, in the following, we assume that our $\lpt$'s satisfy the properties in the following lemma:

\begin{lemma} \label{lem:gapfill} Algorithm \textsf{GapFill} produces a
  modified collection of $\lpt$'s such that
  \begin{OneLiners}
  \item[(i)] For each $i,j, u$ such that $r_u > 0$,
    $\ptime(\head(i,j,u)) \ge 2 \cdot \depth(\head(i,j,u))$.
  \item[(ii)] The total extent of plays at any time $t$, i.e.,
    $\sum_{(i,j,u): \ptime(i,j,u)=t} \prob(i,j,u)$ is at most $3$.
\end{OneLiners}
\end{lemma}

The proof appears in \lref[Appendix]{sec:details-phase-ii}.


\subsubsection{Phase III: Scheduling the Arms}
\label{sec:phase-iii}

Having done the preprocessing, the rounding algorithm is simple: it
first randomly selects at most one strategy forest from the collection
$\{\lpt(i,j)\}_j$ for each arm $i$. It then picks an arm with the
earliest connected component (i.e., that with smallest
$\ptime(\head(i,j,u))$) that contains the current state (the root
states, to begin with), plays it to the end---which either results in
terminating the arm, or making a transition to a state played much later
in time, and repeats. The formal description appears in
\lref[Algorithm]{alg:roundmab}.  (If there are ties in
\lref[Step]{alg:mabstep4}, we choose the smallest index.) Note that the
algorithm runs as long as there is some active node, regardless of
whether or not we have run out of plays (i.e., the budget is
exceeded)---however, we only count the profit from the first $B$ plays
in the analysis.

\newcommand{\activestate}{\mathsf{currstate}}
\begin{algorithm}[ht!]
  \caption{Scheduling the Connected Components: Algorithm \textsf{AlgMAB}}
  \begin{algorithmic}[1]
  \label{alg:roundmab}
    \STATE for arm $i$, \textbf{sample} strategy $\lpt(i,j)$ with
    probability $\frac{\prob(i,j,\rho_i)}{24}$; ignore arm
    $i$ w.p.\ $1 - \sum_{j}
    \frac{\prob(i,j,\rho_i)}{24}$.  \label{alg:mabstep1}
    \STATE let $A \gets$ set of ``active'' arms which chose a
    strategy in the random process. \label{alg:mabstep2}
    \STATE for each $i \in A$, \textbf{let} $\sigma(i) \gets$ index $j$
    of the chosen $\lpt(i,j)$ and \textbf{let} $\activestate(i) \gets $ root
    $\rho_i$. \label{alg:mabstep3}
    \WHILE{active arms $A \neq \emptyset$}
    \STATE \textbf{let} $i^* \gets$ arm with state played earliest in the
    LP (i.e., $i^* \gets \argmin_{i \in A} \{ \ptime(i, \sigma(i), \activestate(i))
    \}$.  \label{alg:mabstep4}
    \STATE \textbf{let} $\tau \gets \ptime(i^*, \sigma(i^*), \activestate(i^*))$.
    \WHILE{$\ptime(i^*, \sigma(i^*), \activestate(i^*)) \neq \infty$
      \textbf{and} $\ptime(i^*, \sigma(i^*), \activestate(i^*)) = \tau$} \label{alg:mabLoop}
    \STATE \textbf{play} arm $i^*$ at state $\activestate(i^*)$ \label{alg:mabPlay}
    \STATE \textbf{update} $\activestate(i^*)$ be the new state of arm
    $i^*$; \textbf{let} $\tau \gets \tau + 1$. \label{alg:mabstep5}
    \ENDWHILE \label{alg:mabEndLoop}
    \IF{$\ptime(i^*, \sigma(i^*), \activestate(i^*)) = \infty$} \label{alg:mabAbandon}
    \STATE \textbf{let} $A \gets A \setminus \{i^*\}$
    \ENDIF
    \ENDWHILE
\end{algorithmic}
\end{algorithm}

Observe that \lref[Steps]{alg:mabLoop}-\ref{alg:mabstep5} play a connected component of a strategy forest contiguously. In particular, this means that all $\activestate(i)$'s considered in \lref[Step]{alg:mabstep4} are head vertices of the corresponding strategy forests. These facts will be crucial in the  analysis.





\begin{lemma} \label{lem:visitprob}
  For arm $i$ and strategy $\lpt(i,j)$, conditioned on $\sigma(i) = j$ after \lref[Step]{alg:mabstep1} of \textsf{AlgMAB},
the probability of playing state $u \in \Si$
  is $\prob(i,j,u)/\prob(i,j,\rho_i)$, where the probability is over the
  random transitions of arm $i$.
\end{lemma}


The above lemma is relatively simple, and proved in
\lref[Appendix]{sec:details-phase-iii}.  The rest of the section proves
that in expectation, we collect a constant factor of the LP reward of
each strategy $\lpt(i,j)$ before running out of budget; the analysis is
inspired by our \sks rounding procedure.  We mainly focus on the
following lemma.
\begin{lemma} \label{lem:beforetime}
  Consider any arm $i$ and strategy $\lpt(i,j)$. Then, conditioned on $\sigma(i)
  = j$ and on the algorithm playing state $u \in \Si$, the probability that this play happens before time
  $\ptime(i,j,u)$ is at least $1/2$.
\end{lemma}

\newcommand{\Evt}{\mathcal{E}}
\newcommand{\bfv}{\mathbf{v}}
\begin{proof}
  Fix an arm $i$ and an index $j$ for the rest of the proof. Given a state $u \in \Si$, let $\Evt_{iju}$
  denote the event $(\sigma(i) = j) \wedge (\text{state $u$ is played})$. Also, let $\bfv = \head(i,j,u)$ be the head of the
  connected component containing $u$ in $\lpt(i,j)$.  Let r.v.\ $\tau_u$
  (respectively $\tau_\bfv$) be the actual time at which state $u$ (respectively state $\bfv$) is played---these random variables take
  value $\infty$ if the arm is not played in these states. Then
  \begin{equation}
     \Pr [ \tau_u \leq \ptime(i,j,u) \mid \Evt_{iju} ] \geq
    \ts \frac{1}{2}
    \iff \Pr [ \tau_\bfv \leq \ptime(i,j,\bfv) \mid \Evt_{iju} ] \geq
\ts \frac{1}{2}\label{eq:7},
  \end{equation}
  because the time between playing $u$ and $\bfv$ is exactly
  $\ptime(i,j,u) - \ptime(i,j,\bfv)$ since the algorithm plays connected components continuously
  (and we have conditioned on $\Evt_{iju}$). Hence, we can just
  focus on proving the right inequality in~(\ref{eq:7}) for vertex $\bfv$.

  For  brevity of notation, let $t_\bfv = \ptime(i,j,\bfv)$. In addition, we define the order $\preceq$ to indicate which states can be played before $\bfv$. That is, again making use of the fact that the algorithm plays connected components contiguously, we say that $(i',j',v') \preceq (i,j,\bfv)$ iff $\ptime(\head(i',j',v')) \le  \ptime(\head(i,j,\bfv))$. Notice that this order is independent of the run of the algorithm.

  For each arm $i' \neq i$ and index $j'$, we define random variables
  $Z_{i'j'}$ used to count the number of plays that can possibly occur
  before the algorithm plays state $\bfv$. If $\mathbf{1}_{(i',j',v')}$ is
  the indicator variable of event $\Evt_{i'j'v'}$, define
  \begin{equation}
    \ts Z_{i',j'} =  \min \big( t_\bfv \; , \; \sum_{v': (i',j',v') \preceq (i,j,\bfv)}
      \mathbf{1}_{(i',j',v')} \big)~.\label{eq:8}
  \end{equation}
  We truncate $Z_{i',j'}$ at $t_\bfv$
  because we just want to capture how much time \emph{up to} $t_\bfv$
  is being used. Now consider the sum $Z = \sum_{i' \neq i} \sum_{j'}
  Z_{i',j'}$.  Note that for arm $i'$, at most one of the $Z_{i',j'}$
  values will be non-zero in any scenario, namely the index
  $\sigma(i')$ sampled in \lref[Step]{alg:mabstep1}. The first claim
  below shows that it suffices to consider the upper tail of~$Z$, and
  show that $\Pr[Z \geq t_\bfv/2] \leq 1/2$, and the second gives a
  bound on the conditional expectation  of $Z_{i',j'}$. 

  \begin{claim} \label{cl:sumbound} $\Pr[ \tau_\bfv \leq t_\bfv \mid \Evt_{iju} ]
    \geq \Pr[ Z \leq t_\bfv/2 ]$.
  \end{claim}

  \begin{proof}
    We first claim that $\Pr[ \tau_\bfv \leq t_\bfv \mid \Evt_{iju} ] \geq
    \Pr[ Z \leq t_\bfv/2 \mid \Evt_{iju}]$. So, let us condition on
    $\Evt_{iju}$.  Then if $Z \leq t_\bfv/2$, none of the $Z_{i',j'}$
    variables were truncated at $t_\bfv$, and hence $Z$ exactly counts
    the total number of plays (by all other arms $i' \neq i$, from any
    state) that could possibly be played before the algorithm plays $v$
    in strategy $\lpt(i,j)$. Therefore, if $Z$ is smaller than $t_\bfv/2$, then
    combining this with the fact that $\depth(v) \leq t_\bfv/2$ (from
    \lref[Lemma]{lem:gapfill}(i)), we can infer that all the plays
    (including those of $v$'s ancestors) that can be made before playing
    $v$ can indeed be completed within $t_\bfv$. In this case the
    algorithm will definitely play $v$ before $t_\bfv$; hence we get
    that conditioning on $\Evt_{iju}$, the event $\tau_\bfv \leq t_\bfv$ holds when $Z \leq
    t_\bfv/2$.

    Finally, to remove the conditioning: note that $Z_{i'j'}$ is just a
    function of (i) the random variables $\mathbf{1}_{(i',j',v')}$, i.e., the random choices made by playing $\lpt(i',j')$, and (ii) the constant $t_\bfv = \ptime(i,j,v)$. However, the r.vs $\mathbf{1}_{(i',j',v')}$ are clearly independent of the event $\Evt_{iju}$ for $i' \neq i$ since the plays of \textsf{AlgMAB} in one arm are independent of the others, and $\ptime(i,j,v)$ is a constant determined once the strategy forests are created in Phase II.  Hence the event $Z \leq t_\bfv/2$ is independent of $\Evt_{iju}$; hence $\Pr[ Z \leq t_\bfv/2 \mid \Evt_{iju}] = \Pr[ Z
    \leq t_\bfv/2]$, which completes the proof.
\end{proof}

  \begin{claim} \label{cl:localexp}
    \[ {\displaystyle \E[ Z_{i',j'} \, | \, \sigma(i') = j'] \leq
      \sum_{v'~\textsf{s.t}~ \ptime(i',j',v') \leq t_\bfv}
      \frac{\prob(i',j',v')}{\prob(i',j',\rho_{i'})} + t_\bfv \left(
        \sum_{v'~\textsf{s.t}~ \ptime(i',j',v') = t_\bfv}
        \frac{\prob(i',j',v')}{\prob(i',j',\rho_{i'})} \right) } \]
  \end{claim}

\begin{proof}
  Recall the definition of $Z_{i'j'}$ in Eq~(\ref{eq:8}): any state $v'$
  with $\ptime(i',j',v') > t_\bfv$ may contribute to the sum only if it
  is part of a connected component with head $\head(i',j',v')$ such that
  $\ptime(\head(i',j',v')) \leq t_\bfv$, by the definition of the
  ordering $\preceq$.  Even among such states, if $\ptime(i',j',v') >
  2t_\bfv$, then the truncation implies that $Z_{i',j'}$ is unchanged
  whether or not we include $\mathbf{1}_{(i',j',v')}$ in the sum.
  Indeed, if $\mathbf{1}_{(i',j',v')} = 1$ then all of $v'$'s ancestors
  will have their indicator variables at value $1$; moreover $\depth(v')
  > t_\bfv$ since there is a contiguous collection of nodes that are
  played from this tree $\lpt(i',j')$ from time $t_\bfv$ onwards till
  $\ptime(i',j',v') > 2t_\bfv$; so the sum would be truncated at value
  $t_\bfv$ whenever $\mathbf{1}_{(i',j',v')} = 1$.
  Therefore, we can write
  \begin{equation}
    \label{eq:12}
    Z_{i',j'} \leq \sum_{v':\ptime(i',j',v') \leq t_\bfv}
    \mathbf{1}_{(i',j',v')} + \sum_{\substack{v': t_\bfv <
        \ptime(i',j',v') \leq 2t_\bfv \\ (i',j',v') \preceq (i,j,v)} }
    \mathbf{1}_{(i',j',v')}
  \end{equation}
  Recall we are interested in the conditional expectation given
  $\sigma(i') = j'$. Note that $\Pr[\mathbf{1}_{(i',j',v')} \mid
  \sigma(i') = j'] = \prob(i',j',v')/\prob(i',j',\rho_{i'})$ by
  \lref[Lemma]{lem:visitprob}, hence the first sum in~(\ref{eq:12})
  gives the first part of the claimed bound. Now the second part:
  observe that for any arm $i'$, any fixed value of $\sigma(i') = j'$,
  and any value of $t' \geq t_\bfv$,
  \[{\displaystyle \sum_{\substack{v'~\textsf{s.t}~\ptime(i',j',v') = t'
        \\ (i',j',v') \preceq (i,j,v)}} \prob(i',j',v') \leq
    \sum_{\substack{ v'~\textsf{s.t}~\ptime(i',j',v') = t_\bfv }}
    \prob(i',j',v') }\]
    This is because of the following argument: Any state that appears on the LHS of the sum above is part of a connected component which crosses $t_\bfv$, they must have an ancestor which is played at $t_\bfv$. Also, since all states which appear in the LHS are played at $t'$, no state can be an ancestor of another. Hence, we can apply the second part of \lref[Observation]{obs:treeflow} and get the above inequality. Combining this with the fact that $\Pr[\mathbf{1}_{(i',j',v')} \mid
  \sigma(i') = j'] = \prob(i',j',v')/\prob(i',j',\rho_{i'})$, and applying it for each value of $t'
  \in (t_\bfv, 2t_\bfv]$, gives us the second term.
\end{proof}

Equipped with the above claims, we are ready to complete the proof of \lref[Lemma]{lem:beforetime}. Employing \lref[Claim]{cl:localexp} we get
  \begin{align}
    \E[ Z ] &=  \sum_{i' \neq i} \sum_{j'} \E[Z_{i',j'} ]
    = \sum_{i' \neq i} \sum_{j'} \E[Z_{i',j'} \mid \sigma(i') = j']\cdot
    \Pr[\sigma(i') = j'] \notag \\
    &= \frac{1}{24} \sum_{i' \neq i} \sum_{j'} \bigg\{
    \sum_{v':\ptime(i',j',v') \leq t_\bfv} \prob(i',j',v') +
    t_\bfv \bigg( \sum_{v': \ptime(i',j',v') = t_\bfv}
    \prob(i',j',v') \bigg) \bigg\} \label{eq:10}\\
    &= \frac{1}{24} \left( 3 \cdot t_\bfv + 3\cdot t_\bfv \right) \leq
    \frac{1}{4} t_\bfv \;. \label{eq:11}
  \end{align}
  Equation~(\ref{eq:10}) follows from the fact that each tree $\lpt(i,j)$ is
  sampled with probability $\frac{\prob(i,j,\rho_i)}{24}$ and (\ref{eq:11}) follows from \lref[Lemma]{lem:gapfill}. Applying Markov's inequality, we have that $\Pr[
  Z \geq t_\bfv/2 ] \leq 1/2$. Finally, \lref[Claim]{cl:sumbound} says
  that $\Pr[ \tau_\bfv \leq t_\bfv \mid \Evt_{iju} ] \geq \Pr[Z \leq
  t_\bfv/2 ] \geq 1/2$, which completes the proof.
\end{proof}

\begin{theorem}
  \label{thm:main-mab}
  The reward obtained by the algorithm~\textsf{AlgMAB} is at least $\Omega(\LPOpt)$.
\end{theorem}

\begin{proof}
The theorem follows by a simple linearity of expectation. Indeed, the expected reward obtained from any state $u \in \Si$ is at least $\sum_{j} \Pr[\sigma(i) = j] \Pr[\textsf{state }~u~\textsf{is played}  \mid  \sigma(i) = j] \Pr [ \tau_u \leq t_u  |  \Evt_{iju}] \cdot R_u \geq \sum_{j} \frac{ \prob(i,j,u)}{24} \frac{1}{2} \cdot R_u$. Here, we have used \lref[Lemmas]{lem:visitprob} and~\ref{lem:beforetime} for the second and third probabilities. But now we can use \lref[Lemma]{lem:convexppt} to infer that $\sum_j \prob(i,j,u) = \sum_t z_{u,t}$; Making this substitution and summing over all states $u \in \Si$ and arms $i$ completes the proof.
\end{proof}

\input{mab-dag}

\section{Concluding Remarks}

We presented the first constant-factor approximations for the
stochastic knapsack problem with cancellations and correlated size/reward
pairs, and for the budgeted learning problem without the martingale
property. We showed that existing LPs for the restricted versions of the
problems have large integrality gaps, which required us to give new LP
relaxations, and well as new rounding algorithms for these problems.

\paragraph*{Acknowledgments.} We thank Kamesh Munagala and Sudipto Guha
for useful conversations.

\bibliographystyle{alpha}
{\small \bibliography{stoc-ks}}

\appendix
\input{bad-egs}
\input{app-knap}

\input{app-polytime}

\input{app-mab}
\input{app-dag}
\input{mab-exploit}




\end{document}

%% file: blurb.tex
\section{Introduction}
\label{sec:introduction}

Stochastic packing problems seem to be conceptually harder than their
deterministic counterparts---imagine a situation where some rounding
algorithm outputs a solution in which the budget constraint has been
exceeded by a constant factor. For deterministic packing problems (with
a single constraint), one can now simply pick the most profitable subset
of the items which meets the packing constraint; this would give us a
profit within a constant of the optimal value. The deterministic packing
problems not well understood are those with multiple (potentially
conflicting) packing constraints.

However, for the stochastic problems, even a single packing constraint
is not simple to handle. Even though they arise in diverse situations,
the first study from an approximations perspective was in an important
paper of Dean et al.~\cite{DeanGV08} (see also~\cite{dgv05,
  Dean-thesis}). They defined the stochastic knapsack problem, where
each job has a random size and a random reward, and the goal is to give
an adaptive strategy for irrevocably picking jobs in order to maximize
the expected value of those fitting into a knapsack with size $B$---they
gave an LP relaxation and rounding algorithm, which produced
\emph{non-adaptive} solutions whose performance was surprisingly within
a constant-factor of the best \emph{adaptive} ones (resulting in a
constant adaptivity gap, a notion they also introduced). However, the
results required that (a)~the random rewards and sizes for items were
independent of each other, and (b)~once a job was placed, it could not
be prematurely canceled---it is easy to see that these assumptions
change the nature of the problem significantly.

The study of the stochastic knapsack problem was very influential---in
particular, the ideas here were used to obtain approximation algorithms
for \emph{budgeted learning problems} studied by Guha and
Munagala~\cite{GuhaM-soda07,GuhaM-stoc07,GuhaM09} and Goel et
al.~\cite{GoelKN09}, among others. They considered problems in the
multi-armed bandit setting with $k$ arms, each arm evolving according to
an underlying state machine with probabilistic transitions when pulled.
Given a budget $B$, the goal is to pull arms up to $B$ times to maximize
the reward---payoffs are associated with states, and the reward is some
function of payoffs of the states seen during the evolution of the
algorithm.  (E.g., it could be the sum of the payoffs of all states
seen, or the reward of the best final state, etc.) The above papers gave
$O(1)$-approximations, index-based policies and adaptivity gaps for
several budgeted learning problems. However, these results all required
the assumption that the rewards satisfied a \emph{martingale property},
namely, if an arm is some state $u$, one pull of this arm would bring an
expected payoff equal to the payoff of state $u$ itself --- the
motivation for such an assumption comes from the fact that the different
arms are assumed to be associated with a fixed (but unknown) reward, but
we only begin with a prior distribution of possible rewards.  Then, the
expected reward from the next pull of the arm, \emph{conditioned} on the
previous pulls, forms a Doob martingale.

However, there are natural instances where the martingale property need
not hold. For instance, the evolution of the prior could not just depend
on the observations made but on external factors (such as time) as
well. Or, in a marketing application, the evolution of a customer's
state may require repeated ``pulls'' (or marketing actions) before the
customer transitions to a high reward state and makes a purchase, while
the intermediate states may not yield any reward.
These lead us to consider the following problem: there are a collection
of $n$ arms, each characterized by an arbitrary (known) Markov chain,
and there are rewards associated with the different states. When we play
an arm, it makes a state transition according to the associated Markov
chain, and fetches the corresponding reward of the new state. What
should our strategy be in order to maximize the expected total reward we
can accrue by making at most $B$ pulls in total?

\subsection{Results} Our main results are the following: We give the first
constant-factor approximations for the general version of the stochastic
knapsack problem where rewards could be correlated with the sizes.
Our techniques are general and also apply to the setting when jobs could be canceled arbitrarily.
We then extend those ideas to give the first
constant-factor approximation algorithms for a class of budgeted learning
problems with Markovian transitions where the martingale property is not satisfied. We summarize these in \lref[Table]{tab:results}.

\begin{table}
\begin{center}
\begin{tabular}{ | l | l | l | l | }
    \hline
   Problem& Restrictions &  Paper  \\ \hline
     Stochastic Knapsack & Fixed Rewards, No Cancellation &  \cite{dgv05} \\ \hline
      & Correlated Rewards, No Cancellation &  \lref[Section]{sec:nopmtn} \\ \hline
      & Correlated Rewards, Cancellation &  \lref[Section]{sec:sk} \\ \hline
     Multi-Armed Bandits & Martingale Assumption &  \cite{GuhaM-soda07} \\ \hline
      & No Martingale Assumption &  \lref[Section]{sec:mab} \\ \hline
    \end{tabular}
 \caption{Summary of Results}\label{tab:results}
\end{center}
\end{table}

\subsection{Why Previous Ideas Don't Extend, and Our Techniques}
\label{sec:high-level-idea}

One reason why stochastic packing problems are more difficult than their
deterministic counterparts is that, unlike in the deterministic setting,
here we cannot simply take a solution with expected reward $R^*$ that
packs into a knapsack of size $2B$ and convert it (by picking a subset
of the items) into a solution which obtains a constant fraction of the
reward $R^*$ whilst packing into a knapsack of size $B$. In fact, there
are examples where a budget of $2B$ can fetch much more reward than what
a budget of size $B$ can (see
\lref[Appendix]{sec:badness-corr}). Another distinction from
deterministic problems is that allowing cancellations can drastically
increase the value of the solution (see
\lref[Appendix]{sec:badness-cancel}). The model used in previous works
on stochastic knapsack and on budgeted learning circumvented both
issues---in contrast, our model forces us to address them.

\textbf{Stochastic Knapsack:} Dean et
al.~\cite{DeanGV08, Dean-thesis} assume that the reward/profit of an
item is independent of its stochastic size. Moreover, their model does
not consider the possibility of canceling jobs in the middle. These assumptions
simplify the structure of the decision tree and make it possible to
formulate a (deterministic) knapsack-style LP, and round it. However,
as shown in \lref[Appendix]{sec:egs}, their LP relaxation performs
poorly when either correlation or cancellation is allowed. This is the first
issue we need to address.

\textbf{Budgeted Learning:} Obtaining approximations for budgeted
learning problems is a more complicated task, since cancellations maybe inherent in the problem formulation, i.e., any strategy would stop playing a particular arm and switch to another, and the rewards by playing any arm are
naturally correlated with the (current) state and hence the number of previous pulls made on the
item/arm. The first issue
is often tacked by using more elaborate LPs with a flow-like structure
that compute a probability distribution over the different times at which the LP stops playing an arm (e.g., \cite{GuhaM-stoc07}), but the latter issue is less understood.
Indeed, several papers on this topic present strategies that fetch an
expected reward which is a constant-factor of an optimal
solution's reward, but which may violate the budget by a constant
factor. In order to obtain an approximate solution without violating the
budget, they critically make use of the \emph{martingale
  property}---with this assumption at hand, they can truncate the last
arm played to fit the budget without incurring any loss in expected
reward.  However, such an idea fails when the martingale property is not
satisfied, and these LPs now have large integrality gaps (see
\lref[Appendix]{sec:badness-corr}).

At a high level, a major drawback with previous LP relaxations for both
problems is that the constraints are \emph{local} for each arm/job,
i.e., they track the probability distribution over how long each
item/arm is processed (either till completion or cancellation), and there
is an additional global constraint binding the total number of
pulls/total size across items. This results in two different issues. For
the (correlated) stochastic knapsack problem, these LPs do not capture
the case when all the items have high contention, since they want to
play early in order to collect profit. And for the general multi-armed
bandit problem, we show that no local LP can be good since such LPs do
not capture the notion of \emph{preempting} an arm, namely switching
from one arm to another, and possibly returning to the original arm
later later. Indeed, we show cases when any near-optimal strategy must
switch between different arms (see
\lref[Appendix]{sec:preemption-gap})---this is a major difference from
previous work with the martingale property where there exist
near-optimal strategies that never return to any arm~\cite[Lemma
2.1]{GuhaM09}. At a high level, the lack of the martingale property
means our algorithm needs to make adaptive decisions, where each move is
a function of the previous outcomes; in particular this may involve
revisiting a particular arm several times, with interruptions in the
middle.  

We resolve these issues in the following manner: incorporating
cancellations into stochastic knapsack can be handled by just adapting
the flow-like LPs from the multi-armed bandits case. To resolve the
problems of contention and preemption, we formulate a \emph{global
  time-indexed} relaxation that forces the LP solution to commit each
job to begin at a time, and places constraints on the maximum expected
reward that can be obtained if the algorithm begins an item a particular
time. Furthermore, the time-indexing also enables our rounding scheme to
extract information about when to preempt an arm and when to re-visit it
based on the LP solution; in fact, these decisions will possibly be
different for different (random) outcomes of any pull, but the LP
encodes the information for each possibility. We believe that our
rounding approach may be of interest in other applications in Stochastic
optimization problems.


Another important version of budgeted learning is when we are allowed to
make up to $B$ plays as usual but now we can ``exploit'' at most $K$ times: reward is only fetched when an arm is exploited and again depends on its current state. There is a further constraint that once an arm is exploited, it must then be discarded.
The LP-based approach here can be easily extended to that case as well.

\subsection{Roadmap}
We begin in \lref[Section]{sec:nopmtn} by presenting a constant-factor approximation algorithm for the stochastic knapsack problem (\sks) when rewards could be correlated with the sizes, but decisions are irrevocable, i.e., job cancellations are not allowed.
Then, we build on these ideas in \lref[Section]{sec:sk}, and present our results
for the (correlated) stochastic knapsack problem, where job cancellation is allowed.

In \lref[Section]{sec:mab}, we move on to the more general class of multi-armed bandit (\mab) problems.
For clarity in exposition, we present our algorithm for \mab, assuming that
the transition graph for each arm is an \emph{arborescence} (i.e., a directed tree), and then
generalize it to arbitrary transition graphs in
\lref[Section]{dsec:mab}.

We remark that while our LP-based approach for
the budgeted learning problem implies approximation algorithms for the
stochastic knapsack problem as well, the knapsack problem provides a
gentler introduction to the issues---it motivates and gives insight into
our techniques for \mab. Similarly, it is easier to understand our techniques for
the \mab problem when the transition graph of each arm's Markov chain is
a tree.
Several illustrative examples are presented in
\lref[Appendix]{sec:egs}, e.g., illustrating why we need adaptive
strategies for the non-martingale \mab problems, and why some natural ideas do not work.
Finally, the extension of our algorithm for \mab for the case when rewards are available only when the arms are explicitly exploited with budgets on both the exploration and exploitation pulls
appear in \lref[Appendix]{xsec:mab}. Note that this algorithm strictly generalizes the previous work on budgeted learning for \mab with the martingale property~\cite{GuhaM-stoc07}. 

\subsection{Related Work}
\label{sec:related-work}

Stochastic scheduling problems have been long studied since the 1960s
(e.g.,~\cite{BirgeL97, Pinedo}); however, there are fewer papers on
approximation algorithms for such problems. Kleinberg et
al.~\cite{KRT-sched}, and Goel and Indyk~\cite{GI99} consider stochastic
knapsack problems with chance constraints: find the max-profit set which
will overflow the knapsack with probability at most $p$. However, their
results hold for deterministic profits and specific size distributions.
Approximation algorithms for minimizing average completion times with
arbitrary job-size distributions was studied by~\cite{MohringSU99,
  SkutU01}. The work most relevant to us is that of Dean, Goemans and
Vondr\'ak~\cite{DeanGV08, dgv05, Dean-thesis} on stochastic knapsack and
packing; apart from algorithms (for independent rewards and sizes), they
show the problem to be PSPACE-hard when correlations are allowed.
\cite{ChawlaR06} study stochastic flow problems. Recent work of Bhalgat
et al.~\cite{BGK11} presents a PTAS but violate the capacity by a factor
$(1+\epsilon)$; they also get better constant-factor approximations
without violations.

The general area of learning with costs is a rich and diverse one (see,
e.g.,~\cite{Bert05,Gittins89}). Approximation algorithms start with the
work of Guha and Munagala~\cite{GuhaM-stoc07}, who gave LP-rounding
algorithms for some problems.  Further papers by these
authors~\cite{GuhaMS07, GuhaM09} and by Goel et al.~\cite{GoelKN09} give
improvements, relate LP-based techniques and index-based policies and
also give new index policies. (See also~\cite{GGM06,GuhaM-soda07}.)
\cite{GuhaM09} considers switching costs, \cite{GuhaMP11} allows pulling
many arms simultaneously, or when there is delayed feedback.  All these
papers assume the martingale condition.


%% file: knap.tex
\newcommand{\er}{\mathsf{ER}}
\newcommand{\skssmall}{{\textsf{StocK-Small}}\xspace}
\newcommand{\skslarge}{{\textsf{StocK-Large}}\xspace}
\newcommand{\sksnocancel}{{\textsf{StocK-NoCancel}}\xspace}

\section{The Correlated Stochastic Knapsack without Cancellation} \label{sec:nopmtn}

We begin by considering the stochastic knapsack problem (\sks), when the job rewards may be correlated with its size.
This generalizes the problem studied by Dean et al. \cite{dgv05} who assume that the rewards are independent of the size of the job.
We first explain why the LP of~\cite{dgv05} has a large integrality gap for our problem; this will naturally motivate our time-indexed formulation.
We then present a simple randomized rounding algorithm which produces a non-adaptive strategy and show that it is an $O(1)$-approximation.

\subsection{Problem Definitions and Notation}
\label{sec:knap-model}

We are given a knapsack of total budget $B$ and a collection of $n$
stochastic items. For any item $i \in [1,n]$, we are given a probability
distribution over $(\mathsf{size}, \mathsf{reward})$ pairs specified as
follows: for each integer value of $t \in [1,B]$, the tuple $(\pi_{i,t},
R_{i,t})$ denotes the probability $\pi_{i,t}$ that item $i$ has a size
$t$, and the corresponding reward is $R_{i,t}$. Note that the reward for a
job is now correlated to its size; however, these quantities for two
different jobs are still independent of each other.

An algorithm to \emph{adaptively} process these items can do the
following actions at the end of each timestep;
\begin{inparaenum}
\item[(i)] an item may complete at a certain size, giving us the
  corresponding reward, and the algorithm may choose a new item to start
  processing, or
\item[(ii)] the knapsack becomes full, at which point the algorithm
  cannot process any more items, and any currently running job does not
  accrue any reward.
\end{inparaenum}
The objective function is to maximize the total expected reward obtained
from all completed items. Notice that we do not allow the algorithm to cancel an item before it completes. We relax this requirement in \lref[Section]{sec:sk}.

\subsection{LP Relaxation}
	The LP relaxation in~\cite{dgv05} was (essentially) a knapsack LP where the sizes of items are replaced by the expected sizes, and the rewards are replaced by the expected rewards. While this was sufficient when an item's reward is fixed (or chosen randomly but independent of its size), we give an example in \lref[Appendix]{sec:badness-corr} where such an LP (and in fact, the class of more general LPs used for approximating \mab problems) would have a large integrality gap. As mentioned in \lref[Section]{sec:high-level-idea}, the reason why local LPs don't work is that there could be high contention for being scheduled early (i.e., there could be a large number of items which all fetch reward if they instantiate to a large size, but these events occur with low probability). In order to capture this contention, we write a global time-indexed LP relaxation.

	The variable $x_{i,t} \in [0,1]$ indicates that item $i$ is
scheduled at (global) time $t$; $S_i$ denotes the random variable for
the size of item $i$, and $\er_{i,t} = \sum_{s
\le B - t} \pi_{i,s} R'_{i,s}$ captures the expected reward that can
be obtained from item $i$ \emph{if it begins} at time $t$; (no reward is obtained for sizes that cannot
fit the (remaining) budget.) 	%
	\begin{alignat}{2} \tag{$\mathsf{LP}_{\sf NoCancel}$} \label{lp:large}
		\max &\ts \sum_{i,t} \er_{i,t} \cdot x_{i,t} &\\
		&\ts \sum_t x_{i,t} \le 1 &\forall i \label{LPbig1}\\
		&\ts \sum_{i, t' \le t} x_{i,t'} \cdot \E[\min(S_i,t)]
                \le 2t  \qquad &\forall t \in [B] \label{LPbig2}\\
		&x_{i,t} \in [0,1] &\forall t \in [B], \forall i \label{LPbig3}
	\end{alignat}	

While the size of the above LP (and the running time of the rounding
algorithm below) polynomially depend on $B$, i.e., pseudo-polynomial, it
is possible to write a compact (approximate) LP and then round it;
details on the polynomial time implementation appear in \lref[Appendix]{app:polytime-nopmtn}.

Notice the constraints involving the \emph{truncated random variables} in equation~\eqref{LPbig2}: these are crucial for showing the correctness of the rounding algorithm \sksnocancel. Furthermore, the ideas used here will appear subsequently in the \mab algorithm later; for \mab, even though we can't explicitly enforce such a constraint in the LP, we will end up inferring a similar family of inequalities from a near-optimal LP solution.

\begin{lemma} \label{thm:lp-large-valid}
The above relaxation is valid for the \sks problem when cancellations are not permitted, and has objective value $\LPOpt \geq \Opt$, where $\Opt$ is the expected profit of an optimal adaptive policy.
\end{lemma}

\begin{proof}
Consider an optimal policy $\Opt$ and let
$x^*_{i,t}$ denote the probability that item $i$ is scheduled
at time $t$. We first show that $\{x^*\}$ is a feasible solution
for the LP relaxation \ref{lp:large}. 	
	It is easy to see that constraints~\eqref{LPbig1} and~\eqref{LPbig3} are
satisfied. To prove that \eqref{LPbig2} are also satisfied, consider
some $t \in [B]$ and some run (over random choices of item sizes) of the optimal policy. Let $\mathbf{1}^{{\sf sched}}_{i,t'}$ be indicator variable
that item $i$ is scheduled at time $t'$ and let
$\mathbf{1}^{{\sf size}}_{i,s}$ be the indicator variable for whether the size of item $i$
is $s$. Also, let $L_t$ be the random variable indicating the
last item scheduled at or before time $t$. Notice that $L_t$
is the only item scheduled before or at time $t$ whose
execution may go over time $t$. Therefore, we get that
$$\sum_{i \neq L_t} \sum_{t' \le t} \sum_{s\leq B} \mathbf{1}^{{\sf sched}}_{i,t'} \cdot
\mathbf{1}^{{\sf size}}_{i,s} \cdot s   \le t.$$ Including $L_t$ in the summation and truncating the sizes by $t$, we immediately obtain $$\sum_i \sum_{t' \le t} \sum_{s} \mathbf{1}^{{\sf sched}}_{i,t'} \cdot
\mathbf{1}^{{\sf size}}_{i,s} \cdot \min(s, t) \le 2t.$$ Now, taking expectation (over all of \Opt's sample paths) on
both sides and using linearity of expectation we have $$\sum_i
\sum_{t' \le t} \sum_{s} \E \left[\mathbf{1}^{{\sf sched}}_{i,t'} \cdot \mathbf{1}^{{\sf size}}_{i,s}\right]
\cdot \min(s,t) \le 2t.$$

However, because $\Opt$ decides whether to schedule an item before observing the size it instantiates to, we have that
$\mathbf{1}^{{\sf sched}}_{i,t'}$ and $\mathbf{1}^{{\sf size}}_{i,s}$ are independent random variables; hence, the LHS above can be re-written as
	\begin{align*}
		&\sum_i \sum_{t' \le t} \sum_s \Pr[\mathbf{1}^{{\sf sched}}_{i,t'} = 1
\wedge \mathbf{1}^{{\sf size}}_{i,s} = 1] \min(s,t) \\
		 &= \sum_i \sum_{t' \le t} \Pr[\mathbf{1}^{{\sf sched}}_{i,t'} = 1]
\sum_s \Pr[\mathbf{1}^{{\sf size}}_{i,s} = 1] \min(s,t) \\
		 &= \sum_i \sum_{t' \le t} x^*_{i,t'} \cdot
\E[\min(S_i,t)]
	\end{align*}
	Hence constraints \eqref{LPbig2} are satisfied. 	
	Now we argue that the expected reward of $\Opt$ is equal to
the value of the solution $x^*$. Let $O_i$ be the random variable denoting the reward
obtained by $\Opt$ from item $i$. Again, due to the independence
between  $\Opt$ scheduling an item and the size it instantiates to, we get
that the expected reward that $\Opt$ gets from executing item
$i$ at time $t$ is $$\E[O_i | \mathbf{1}^{{\sf sched}}_{i,t} = 1] = \sum_{s
\le B - t} \pi_{i,s} R_{i,s} = \er_{i,t}.$$ Thus the expected
reward from item $i$ is obtained by considering all possible
starting times for $i$:
	\begin{align*}
		\E[O_i] = \sum_t \Pr[\mathbf{1}^{{\sf sched}}_{i,t} = 1] \cdot \E[O_i |
\mathbf{1}^{{\sf sched}}_{i,t} = 1] = \sum_t \er_{i,t} \cdot x^*_{i,t}.
	\end{align*} 	
	This shows that \ref{lp:large} is a valid relaxation for
our problem and completes the proof of the lemma. 	
\end{proof}

We are now ready to present our rounding algorithm \sksnocancel (\lref[Algorithm]{alg:sksnocancel}). It a simple randomized rounding procedure which (i)  picks the start time of each item according to the corresponding distribution in the optimal LP solution, and (ii) plays the items in order of the (random) start times. To ensure that the budget is not violated, we also drop each item independently with some constant probability.
\begin{algorithm}[ht!]
  \caption{Algorithm \sksnocancel}
  \begin{algorithmic}[1]
  \label{alg:sksnocancel}

    \STATE for each item $i$, \textbf{assign} a random start-time $D_i = t$ with
    probability $\frac{x^*_{i,t}}{4}$; with probability $1 - \sum_{t} \frac{x^*_{i,t}}{4}$, completely ignore item $i$ ($D_i = \infty$ in this case).  \label{alg:big1}


    \FOR{$j$ from $1$ to $n$}
    	\STATE Consider the item $i$ which has the $j$th smallest deadline (and $D_i \neq \infty$) \label{alg:big2}

    	\IF{the items added so far to the knapsack occupy at most $D_i$ space}
    		\STATE add $i$ to the knapsack. \label{alg:big3}
    	\ENDIF \label{alg:big4}
    \ENDFOR
\end{algorithmic}
\end{algorithm}

	Notice that the strategy obtained by the rounding procedure
obtains reward from all items which are not dropped and which
do not fail (i.e. they can start being scheduled before the sampled start-time $D_i$ in \lref[Step]{alg:big1}); we now bound the failure probability.
	\begin{lemma} \label{lem:big-fail}
		For every $i$, $\Pr(i~\mathsf{fails} \mid  D_i = t) \le 1/2$.
	\end{lemma} 	
	\begin{proof}
		Consider an item $i$ and time $t \neq \infty$ and condition on the event that $D_i = t$.
Let us consider
the execution of the algorithm when it tries to add item $i$ to
the knapsack in \lref[steps]{alg:big2}-\ref{alg:big4}. Now, let
$Z$ be a random variable denoting \emph{how much of the interval} $[0,t]$ of
the knapsack is occupied  by previously scheduling items, at the time when $i$ is considered for addition; since $i$ does not fail
when $Z < t$, it suffices to prove that $\Pr(Z \ge t) \le 1/2$.
		
For some item $j \neq i$, let $\mathbf{1}_{D_j \le t}$ be the indicator variable that $D_j \le t$;
notice that by the order in which algorithm \sksnocancel adds items into the knapsack, it is also the indicator that $j$ was considered before $i$.
In addition, let $\mathbf{1}^{{\sf size}}_{j,s}$ be the indicator variable that $S_j
= s$. Now, if $Z_j$ denotes the total amount of
the interval $[0,t]$ that that $j$ occupies, we have
$$ Z_j \le \mathbf{1}_{D_j \le t}
\sum_s \mathbf{1}^{{\sf size}}_{j,s} \min(s, t).$$
 Now, using the independence of
$\mathbf{1}_{D_j \le t}$ and $\mathbf{1}^{{\sf size}}_{j,s}$, we have
		%
		\begin{equation}
			\E[Z_j] \ts \le \E[\mathbf{1}_{D_j \le t}] \cdot \E[\min(S_j,
t)] =  \frac{1}{4} \sum_{t' \le t} x^*_{j,t'} \cdot
\E[\min(S_j, t)]
		\end{equation}
		Since $Z = \sum_j Z_j$, we can use linearity of expectation
and the fact that $\{x^*\}$ satisfies LP constraint~\eqref{LPbig2} to get
		\begin{align*}
			\E[Z] &\ts \le \frac{1}{4} \sum_j \sum_{t' \le t}
x^*_{j,t'} \cdot \E[\min(S_j, t)] \le \frac{t}{2}\;.
		\end{align*}
		To conclude the proof of the lemma, we apply Markov's
inequality to obtain $\Pr(Z \ge t) \le 1/2$.
	\end{proof}
To complete the analysis, we use the fact that any item chooses a random start time $D_i = t$ with probability $x^*_{i,t}/4$, and conditioned on this event, it is added to the knapsack with probability at least $1/2$ from \lref[Lemma]{lem:big-fail}; in this case, we get an expected reward of at least $\er_{i,t}$. The theorem below (formally proved in \lref[Appendix]{app:nopmtn-proof} then follows by linearity of expectations.

		\begin{theorem}\label{thm:large}
		The expected reward of our randomized algorithm is at
least $\frac18$ of $\LPOpt$.
	\end{theorem}

\section{Stochastic Knapsack with Correlated Rewards and Cancellations}
\label{sec:sk}

In this section, we present our algorithm for stochastic knapsack (\sks)
where we allow correlations between rewards and sizes, and also allow
cancellation of jobs.
%
The example in \lref[Appendix]{sec:badness-cancel} shows that there can be an arbitrarily large gap in the expected profit between strategies that can cancel jobs and those that can't. Hence we need to write new LPs to capture the benefit of cancellation, which we do in the following manner.

Consider any job $j$: we can create two jobs from it, the ``early''
version of the job, where we discard profits from any instantiation
where the size of the job is more than $B/2$, and the ``late'' version
of the job where we discard profits from instantiations of size
at most $B/2$. Hence, we can get at least half the optimal value by
flipping a fair coin and either collecting rewards from either the early
or late versions of jobs, based on the outcome. In the next section,
we show how to obtain a constant factor approximation for the first
kind. For the second kind, we argue that cancellations don't help; we can then reduce it to \sks without cancellations (considered in \lref[Section]{sec:nopmtn}).



\subsection{Case I: Jobs with Early Rewards} \label{caseSmall}

We begin with the setting in which only small-size instantiations of items may fetch reward, i.e., the rewards $R_{i,t}$ of every item $i$ are assumed to be $0$ for $t > B/2$. 
In the following LP relaxation \ref{lpone}, $v_{i,t} \in [0,1]$ tries to
capture the probability with which $\Opt$ will process item $i$ for
\emph{at least} $t$ timesteps\footnote{In the following two sections, we
  use the word timestep to refer to processing one unit of some item.},
$s_{i,t} \in [0,1]$ is the probability that $\Opt$ stops processing item
$i$ \emph{exactly} at $t$ timesteps. The time-indexed formulation causes
the algorithm to have running times of $\poly(B)$---however, it is easy
to write compact (approximate) LPs and then round them; we describe the
necessary changes to obtain an algorithm with running time $\poly(n,
\log B)$ in \lref[Appendix]{app:polytime}.
\begin{alignat}{2}
  \max &\ts \sum_{1 \leq t \leq B/2} \sum_{1 \leq i \leq n}  v_{i,t} \cdot R_{i,t}  \frac{\pi_{i,t}}{\sum_{t' \geq t} \pi_{i,t'}}   & &
  \tag{$\mathsf{LP}_S$} \label{lpone} \\
  & v_{i,t} = s_{i,t} + v_{i,t+1} & \qquad & \forall \,
  t \in [0,B], \, i \in [n]   \label{eq:1} \\
  &s_{i,t} \geq  \frac{\pi_{i,t}}{\sum_{t' \geq t} \pi_{i,t'}} \cdot v_{i,t} & \qquad & \forall \, t \in
  [0,B], \, i \in [n] \label{eq:2} \\
  &\ts \sum_{i \in [n]} \sum_{t \in [0,B]} t \cdot s_{i,t}  \leq B
    &   \label{eq:3}\\
  &v_{i,0} = 1 & \qquad & \forall \, i \label{eq:4} \\
  v_{i,t}, s_{i,t} &\in [0,1] & \qquad & \forall \, t \in [0,B], \, i \in
  [n] \label{eq:5}
\end{alignat}

\begin{theorem}
  \label{thm:lp1-valid}
  The linear program~(\ref{lpone}) is a valid relaxation for the \sks
  problem, and hence the optimal value $\LPOpt$ of the LP  is at least
  the total expected reward $\Opt$ of an optimal solution.
\end{theorem}

\begin{proof}
  Consider an optimal solution $\Opt$ and let $v^*_{i,t}$
  and $s^*_{i,t}$ denote the probability that $\Opt$ processes item $i$
  for at least $t$ timesteps, and the probability that $\Opt$ stops
  processing item $i$ at exactly $t$ timesteps.  We will now show that all the constraints of ~\ref{lpone} are satisfied one by one.

To this end, let $R_i$ denote the random variable (over different executions of $\Opt$) for the amount of processing done on job $i$.
Notice that $ \Pr[R_i \geq t] = \Pr[R_i \geq (t+1)]  + \Pr[R_i = t]$.
But now, by definition we have $\Pr[R_i \geq t] = v^*_{i,t}$ and $\Pr[R_i = t] = s^*_{i,t}$.
This shows that $\{v^*, s^*\}$ satisfies these constraints.

  For the next constraint, observe that conditioned on $\Opt$ running an item
  $i$ for at least $t$ time steps, the probability of item $i$ stopping due to its size having
  instantiated to exactly equal to $t$ is $\pi_{i,t}/\sum_{t' \geq t}
  \pi_{i,t'}$, i.e.,
 $\Pr [ R_i = t \mid R_i \geq t ] \geq \pi_{i,t}/\sum_{t' \geq t} \pi_{i,t'}$.
 This shows that $\{v^*, s^*\}$ satisfies constraints~(\ref{eq:2}).

  Finally, to see why constraint~(\ref{eq:3}) is satisfied, consider any
  particular run of the optimal algorithm and let $\mathbf{1}^{stop}_{i,t}$
  denote the indicator random variable of the event $R_i = t$.
Then we have
  \[ \sum_{i} \sum_{t} \mathbf{1}^{stop}_{i,t} \cdot t \leq B \]
  Now, taking expectation over all runs of $\Opt$ and using linearity of
  expectation and the fact that $\E[\mathbf{1}^{stop}_{i,t}] = s^*_{i,t}$, we
  get constraint~(\ref{eq:3}).  As for the objective function, we again
  consider a particular run of the optimal algorithm and let
  $\mathbf{1}^{proc}_{i,t}$ now denote the indicator random variable for the event $(R_i \geq t)$,
and $\mathbf{1}^{size}_{i,t}$ denote the indicator variable for whether the size
  of item $i$ is instantiated to exactly $t$ in this run. Then we have the total
  reward collected by $\Opt$ in this run to be exactly
  \[ \sum_{i} \sum_{t} \mathbf{1}^{proc}_{i,t} \cdot \mathbf{1}^{size}_{i,t} \cdot
  R_{i,t} \]
  Now, we simply take the expectation of the above random variable over
  all runs of $\Opt$, and then use the following fact about
  $\E[\mathbf{1}^{proc}_{i,t} \mathbf{1}^{size}_{i,t}]$:
  \begin{eqnarray}
    \nonumber \E[\mathbf{1}^{proc}_{i,t} \mathbf{1}^{size}_{i,t}] &=& \Pr[\mathbf{1}^{proc}_{i,t} = 1 \wedge \mathbf{1}^{size}_{i,t} = 1]\\
    \nonumber & =&  \Pr[\mathbf{1}^{proc}_{i,t} = 1] \Pr[\mathbf{1}^{size}_{i,t} = 1 \, |\, \mathbf{1}^{proc}_{i,t} = 1] \\
    \nonumber & =& v^*_{i,t} \frac{\pi_{i,t}}{\sum_{t' \geq t} \pi_{i,t'}}
  \end{eqnarray}
  We thus get that the expected reward collected by $\Opt$ is exactly
  equal to the objective function value of the LP formulation for the
  solution $(v^*, s^*)$.
\end{proof}


Our rounding algorithm is very natural, and simply tries
to mimic the probability distribution (over when to stop each
item) as suggested by the optimal LP solution. To this end, let
$(v^*, s^*)$ denote an optimal fractional solution. The reason
why we introduce some damping (in the selection probabilities)
up-front is to make sure that we could appeal to Markov's inequality and ensure that the knapsack does not get
violated with good probability.

\begin{algorithm}[ht!]
  \caption{Algorithm \skssmall}
  \begin{algorithmic}[1]
  \label{alg:skssmall}
    \FOR{each item $i$}

    \STATE \textbf{ignore} $i$ with probability $1-1/4$ (i.e., do not schedule it at all). \label{alg:st:1}

    \FOR{$0 \leq t \leq B/2$}

    \STATE \textbf{cancel} item $i$ at this step with probability $\frac{s^*_{i,t}}{v^*_{i,t}} - \frac{\pi_{i,t}}{\sum_{t' \geq t} \pi_{i,t'}}$ and \textbf{continue} to next item.  \label{alg:st:2}


    \STATE process item $i$ for its $(t+1)^{st}$ timestep.  \label{alg:st:4}


    \IF{item $i$ terminates after being processed for exactly $(t+1)$ timesteps}

    \STATE \textbf{collect} a reward of $R_{i,t+1}$ from this item; \textbf{continue} onto next item;  \label{alg:st:5}

    \ENDIF

    \ENDFOR

    \ENDFOR
\end{algorithmic}
\end{algorithm}

Notice that while we let the algorithm proceed even if its budget is violated, we will
collect reward only from items that complete before time $B$. This simplifies the analysis a fair bit, both here and for the \mab algorithm.
In \lref[Lemma]{lem:stop-dist} below (proof in \lref[Appendix]{app:small}), we show that for any item that is not dropped in \lref[step]{alg:st:1}, its probability distribution over
stopping times is identical to the optimal LP solution $s^*$. We then use this to argue that the expected reward of our algorithm is $\Omega(1)\LPOpt$.
\begin{lemma} \label{lem:stop-dist}
Consider item $i$ that was not dropped in
\lref[step]{alg:st:1},
Then, for any timestep $t \geq 0$, the following hold:
\begin{OneLiners}
\item[(i)] The probability (including cancellation\&
  completion) of stopping at timestep $t$ for item $i$ is $s^*_{i,t}$.
\item[(ii)] The probability that item $i$ gets processed for
its
  $(t+1)^{st}$ timestep is exactly $v^*_{i,t+1}$
\item[(iii)] If item $i$ has been processed for $(t+1)$ timesteps,
  the probability of completing successfully at timestep $(t+1)$ is
  $\pi_{i,t+1}/\sum_{t' \geq t +1} \pi_{i,t'}$
\end{OneLiners}
\end{lemma}

\begin{theorem} \label{thm:small}
  The expected reward of our randomized algorithm is at least $\frac18$ of $\LPOpt$.
\end{theorem}

\begin{proof}
  Consider any item $i$. In the worst case, we process it after all other items. Then the total expected size occupied thus far is at most
  $\sum_{i' \neq i} \mathbf{1}^{keep}_{i'} \sum_{t \geq 0} t \cdot s^*_{i',t}$,
  where $\mathbf{1}^{keep}_{i'}$ is the indicator random variable denoting
  whether item $i'$ is not dropped in \lref[step]{alg:st:1}. Here we have used \lref[Lemma]{lem:stop-dist} to argue that if an item $i'$ is selected, its stopping-time distribution follows $s^*_{i',t}$.
 Taking expectation over the randomness in \lref[step]{alg:st:1}, the expected space occupied by other jobs is at most $\sum_{i'
    \neq i} \frac{1}{\factor} \sum_{t \geq 0} t \cdot s^*_{i',t} \leq
  \frac{B}{4}$.  Markov's inequality implies that this is at most
  $B/2$ with probability at least $1/2$. In this case, if item $i$ is started (which happens w.p. $1/4$), it runs without violating the knapsack, with expected reward $\sum_{t \geq 1} v^*_{i,t} \cdot \pi_{i,t}/(\sum_{t' \geq t}
  \pi_{i,t'})$; the total expected reward is then at least $\sum_{i} \frac{1}{8} \sum_{t} v^*_{i,t}
  \pi_{i,t}/(\sum_{t' \geq t} \pi_{i,t'}) \geq \frac{\LPOpt}{8}$.
\end{proof}


\subsection{Case II: Jobs with Late Rewards} \label{sec:large}

Now we handle instances in which only large-size instantiations of items may fetch reward, i.e., the rewards $R_{i,t}$ of every item $i$ are assumed to be $0$ for $t \leq B/2$.
For such instances, we now argue that \emph{cancellation is not helpful}. As a consequence, we can use the results of \lref[Section]{sec:nopmtn} and obtain a constant-factor approximation algorithm!

To see why, intuitively, as an algorithm processes a job for its $t^{th}$ timestep for $t < B/2$, it gets no more information about the reward than when starting (since all rewards are at large sizes).
Furthermore, there is no benefit of canceling a job once it has run for at least $B/2$ timesteps -- we can't get any reward by starting some other item.

More formally, consider a
(deterministic) strategy $S$ which in some state makes the
decision of scheduling item $i$ and halting its execution if it
takes more than $t$ timesteps. First suppose that $t \le B/2$;
since this job does will not be able to reach size larger than
$B/2$, no reward will be accrued from it and hence we can
change this strategy by skipping the scheduling of $i$ without
altering its total reward. Now consider the case where $t >
B/2$. Consider the strategy $S'$ which behaves as $S$ except
that it does not preempt $i$ in this state but lets $i$ run to
completion. We claim that $S'$ obtains at least as much
expected reward as $S$. First, whenever item $i$ has size at
most $t$ then $S$ and $S'$ obtain the same reward. Now suppose
that we are in a scenario where $i$ reached size $t > B/2$.
Then item $i$ is halted and $S$ cannot obtain any other
reward in the future, since no item that can fetch any reward would complete before the budget runs out; in the same
situation, strategy $S'$ obtains non-negative rewards. 	
Using this argument we can eliminate all the cancellations
of a strategy without decreasing its expected reward. 	
	\begin{lemma}
		There is an optimal solution in this case which does not cancel.
	\end{lemma}

As mentioned earlier, we can now appeal to the results of \lref[Section]{sec:nopmtn} and obtain a constant-factor approximation for the large-size instances. Now we can combine the algorithms that handle the two different scenarios (or choose one at random and run it), and get a constant fraction of the expected reward that an optimal policy fetches.



%% file: mab-dag.tex
\newcommand{\lpd}{\mathbb{D}}
\newcommand{\lpdt}{\mathbb{DT}}
\newcommand{\state}{\mathsf{state}}
\newcommand{\mroot}{\mathsf{root}}
\newcommand{\currnode}{\mathsf{currnode}}

\section{MABs with Arbitrary Transition Graphs}
\label{dsec:mab}

We now show how we can use techniques akin to those we described for the
case when the transition graph is a tree, to handle the case when it can
be an arbitrary directed graph. A na\"{\i}ve way to do this is to expand
out the transition graph as a tree, but this incurs an exponential
blowup of the state space which we want to avoid. We can assume we
have a layered DAGs, though, since the conversion from a digraph to a
layered DAG only increases the state space by a factor of the horizon
$B$; this standard reduction appears in
\lref[Appendix]{dsec:layered-enough}.

While we can again write an LP relaxation of the problem for layered DAGs, the
challenge arises in the rounding algorithm: specifically, in (i)
obtaining the convex decomposition of the LP solution as in Phase~I, and
(ii) eliminating small gaps as in Phase~II by advancing forests in the
strategy.
\begin{itemize}
\item We handle the first difficulty by considering convex
  decompositions not just over strategy forests, but over slightly
  more sophisticated strategy DAGs. Recall (from \lref[Figure]{fig:treeforest}) that in the tree case, each
  state in a strategy forest was labeled by a unique time and a unique
  probability associated with that time step. As the name suggests, we
  now have labeled DAGs---but the change is more than just that.  Now
  each state has a copy associated with \emph{each} time step in
  $\{1, \ldots, B\}$. This change tries to capture the fact that our strategy may
  play from a particular state $u$ at different times depending on the
  path taken by the random transitions used to reach this state.  (This
  path was unique in the tree case.)

\item Now having sampled a strategy DAG for each arm, one can expand
  them out into strategy forests (albeit with an exponential blow-up in the size), and use Phases~II and~III from our
  previous algorithm---it is not difficult to prove that this algorithm is a constant-factor
  approximation. However, the above is not a poly-time algorithm, since the size of the strategy forests may be exponentially large. If we don't expand the DAG, then we do not see how to
  define gap elimination for Phase~II. But we observe that instead of
  explicitly performing the advance steps in Phase~II, it suffices to
  perform them as a \emph{thought experiment}---i.e., to not alter the
  strategy forest at all, but merely to infer when these advances would
  have happened, and play accordingly in the Phase~III~\footnote{This is similar to the idea of lazy evaluation of strategies. The DAG contains an implicit randomized strategy which we make explicit as we toss coins of the various outcomes using an algorithm.}.  Using this, we
  can give an algorithm that plays just on the DAG, and argue that the
  sequence of plays made by our DAG algorithm faithfully mimics the
  execution if we had constructed the exponential-size tree from the
  DAG, and executed Phases~II and~III on that tree.
\end{itemize}
The details of the LP rounding algorithm for layered DAGs follows in
\lref[Sections]{dsec:lp-dag}-\ref{dsec:phase-iii}.

\subsection{LP Relaxation} \label{dsec:lp-dag}

There is only one change in the LP---constraint~\eqref{eq:mabdaglp1} now
says that if a state $u$ is visited at time $t$, then one of its
ancestors must have been pulled at time $t-1$; this ancestor was
unique in the case of trees.
\begin{alignat}{2} \tag{$\mathsf{LP}_\mathsf{mabdag}$} \label{lp:mabdag}
	\max \ts \sum_{u,t} r_u &\cdot z_{u,t}\\
	w_{u,t} &= \sum_{v} z_{v, t-1}  \cdot p_{v,u} & \qquad \forall t \in [2,B],\, u \in \S \setminus \cup_{i} \{\rho_i\},\, v \in \S \label{eq:mabdaglp1}\\
	\ts \sum_{t' \le t} w_{u,t'} &\geq \ts \sum_{t' \leq t} z_{u,t'} & \qquad \forall t \in [1,B], \, u \in \S \label{eq:mabdaglp2}\\
	\ts \sum_{u \in \S} z_{u,t} &\le 1 & \qquad \forall t \in [1,B]  \label{eq:mabdaglp3}\\
	w_{\rho_i, 1} &= 1 & \qquad \forall i \in [1,n] \label{eq:mabdaglp4}
\end{alignat}	
Again, a similar analysis to the tree case shows that this is a valid
relaxation, and hence the LP value is at least the optimal expected reward.

\subsection{Convex Decomposition: The Altered Phase~I}
\label{dsec:phase-i}

This is the step which changes the most---we need to incorporate the
notion of peeling out a ``strategy DAG'' instead of just a tree.  The
main complication arises from the fact that a play of a state $u$ may
occur at different times in the LP solution, depending on the path to
reach state $u$ in the transition DAG. However, we don't need to keep
track of the entire history used to reach $u$, just how much time has
elapsed so far. With this in mind, we create $B$ copies of each state
$u$ (which will be our nodes in the strategy DAG), indexed by $(u,t)$ for $1 \leq t \leq B$.

The $j^{th}$ \emph{strategy dag} $\lpd(i,j)$ for arm $i$ is an
assignment of values $\prob(i,j,u,t)$ and a relation `$\rightarrow$'
from 4-tuples to 4-tuples of the form $(i,j,u,t) \rightarrow (i,j,v,t')$
such that the following properties hold:
\begin{OneLiners}
\item[(i)] For $u,v \in \Si$ such that $p_{u,v} > 0$ and any time $t$,
  there is exactly one time $t' \geq t+1$ such that $(i,j,u,t)
  \rightarrow (i,j,v,t')$. Intuitively, this says if the arm is played
  from state $u$ at time $t$ and it transitions to state $v$, then it is
  played from $v$ at a unique time $t'$, if it played at all. If $t' =
  \infty$, the play from $v$ never happens.
\item[(ii)] For any $u \in \Si$ and time $t \neq \infty$, $\prob(i,j,u,t) = \sum_{(v,t')~\mathsf{s.t}~(i,j,v,t')\rightarrow(i,j,u,t)} \prob(i,j,v,t') \cdot p_{v,u}$.
\end{OneLiners}

For clarity, we use the following notation throughout the remainder of the section: \emph{states} refer to the states in the original transition DAG, and \emph{nodes} correspond to the tuples $(i,j,u,t)$ in the strategy DAGs. When $i$ and $j$ are clear in context, we may simply refer to a node of the strategy DAG by $(u,t)$.


Equipped with the above definition, our convex decomposition procedure
appears in \lref[Algorithm]{dalg:dconvex}. The main subroutine involved
is presented first~(\lref[Algorithm]{dalg:convex-sub}). This subroutine,
given a fractional solution, identifies the structure of the DAG that
will be peeled out, depending on when the different states are first
played fractionally in the LP solution. Since we have a layered DAG, the
notion of the \emph{depth} of a state is well-defined as the number of
hops from the root to this state in the DAG, with the depth of the root
being $0$.

\newcommand{\peelStrat}{\textsf{PeelStrat}\xspace}
\newcommand{\pprob}{\mathsf{peelProb}}

\begin{algorithm}[ht!]
\caption{Sub-Routine \peelStrat(i,j)}
\begin{algorithmic}[1]
\label{dalg:convex-sub}
	\STATE {\bf mark} $(\rho_i,t)$ where $t$ is the earliest time s.t.\ $z_{\rho_i,t} > 0$ and set $\pprob(\rho_i,t) = 1$. All other nodes are un-marked and have $\pprob(v,t') = 0$.
	\WHILE {$\exists$ a marked unvisited node}
        \STATE {\bf let} $(u,t)$ denote the marked node of smallest depth and earliest time; {\bf update} its status to visited.
	\FOR {every $v$ s.t.\ $p_{u,v} > 0$}
		\IF{there is $t'$ such that $z_{v,t'} > 0$, consider the earliest such $t'$ and}
			\STATE {\bf mark} $(v,t')$ and {\bf set}
                        $(i,j,u,t) \rightarrow (i,j,v,t')$; {\bf update} $\pprob(v,t') := \pprob(v,t') + \pprob(u,t)\cdot p_{u,v}$. \label{dalg:peel3}
		\ELSE
			\STATE {\bf set} $(i,j,u,t) \rightarrow (i,j,v,\infty)$ and leave $\pprob(v,\infty) = 0$.
		\ENDIF
	\ENDFOR
	\ENDWHILE
\end{algorithmic}
\end{algorithm}

The convex decomposition algorithm is now very easy to describe with the sub-routine in \lref[Algorithm]{dalg:convex-sub} in hand.

\begin{algorithm}[ht!]
\caption{Convex Decomposition of Arm $i$}
\begin{algorithmic}[1]
\label{dalg:dconvex}
	\STATE {\bf set} ${\cal C}_i \leftarrow \emptyset$ and {\bf set loop index} $j \leftarrow 1$.
	\WHILE {$\exists$ a state $u \in \Si$ s.t.\ $\sum_{t} z^{j-1}_{u,t} > 0$} \label{dalg:convex1}
	\STATE {\bf run} sub-routine \peelStrat to extract a DAG $\lpd(i,j)$ with the appropriate $\pprob(u,t)$ values.
	\STATE {\bf let} $A \leftarrow \{(u,t)~\mathsf{s.t}~\pprob(u,t) \neq 0\}$.
	\STATE {\bf let} $\epsilon = \min_{(u,t) \in A} z^{j-1}_{u,t}/\pprob(u,t)$. \label{dalg:convex3}
	\FOR{ every $(u,t)$} \label{dalg:convex3a}
		\STATE {\bf set} $\prob(i,j,u,t) = \epsilon \cdot \pprob(u,t)$. \label{dalg:convex4}
		\STATE {\bf update} $z^j_{u,t} = z^{j-1}_{u, t} - \prob(i,j,u,t)$. \label{dalg:convex5}
		\STATE {\bf update} $w^j_{v, t+1} = w^{j-1}_{v, t+1} - \prob(i,j,u,t) \cdot p_{u,v}$ for all $v$. \label{dalg:convex6}
	\ENDFOR	
	\STATE {\bf set} ${\cal C}_i \leftarrow {\cal C}_i \cup \lpd(i,j)$.  	 \label{dalg:convex7}
	\STATE {\bf increment} $j \leftarrow  j + 1$.
	\ENDWHILE
\end{algorithmic}
\end{algorithm}

An illustration of a particular DAG and a strategy dag $\lpd(i,j)$
peeled off is given in \lref[Figure]{dfig:dag} (notice that the states $w$, $y$ and $z$ appear more than once depending on the path taken to reach them).

\begin{figure}[ht]
\centering
\subfigure[DAG for some arm $i$]{
\includegraphics[scale=0.7]{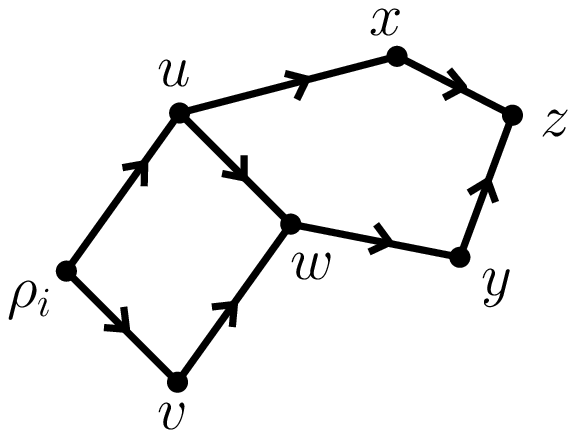}
\label{dfig:subfig1}
}
\hspace{20pt}
\subfigure[Strategy dag $\lpd(i,j)$]{
\includegraphics[scale=0.5]{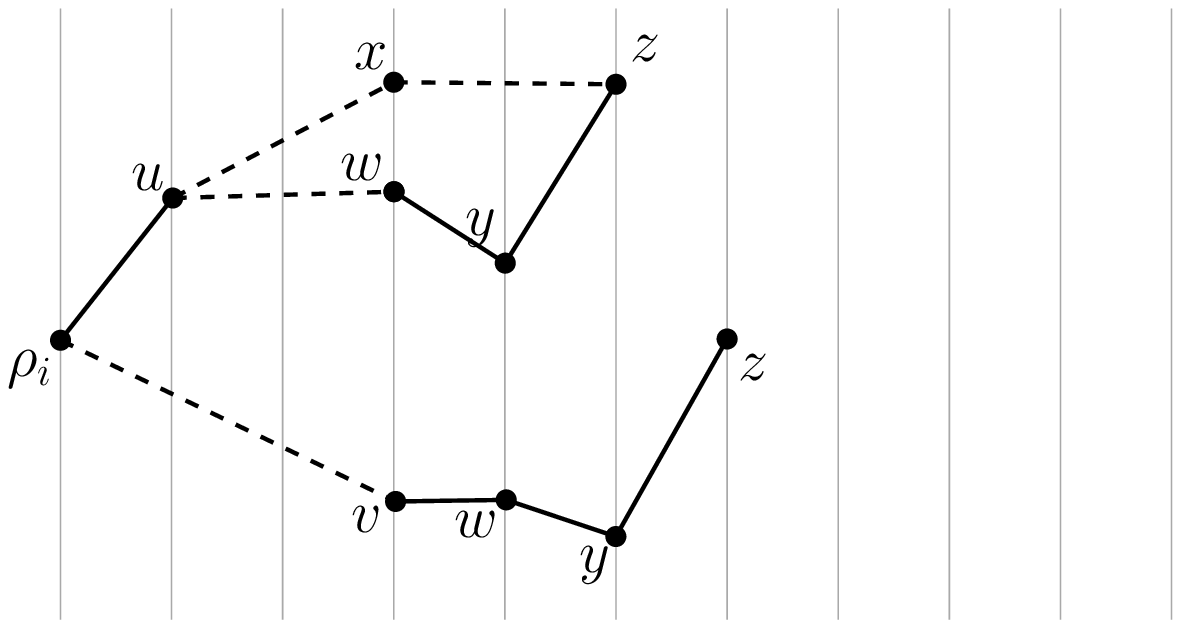}
\label{dfig:subfig2}
}
\caption{Strategy dags and how
  to visualize them: notice the same state played at different times.}
\label{dfig:dag}
\end{figure}

Now we analyze the solutions $\{z^j, w^j\}$ created by \lref[Algorithm]{dalg:dconvex}.

\begin{lemma} \label{dlem:convexstep} Consider an integer $j$ and
  suppose that $\{z^{j-1}, w^{j-1}\}$ satisfies
  constraints~\eqref{eq:mablp1}-\eqref{eq:mablp3} of
  \ref{lp:mabdag}. Then after iteration $j$ of \lref[Step]{dalg:convex1},
  the following properties hold:
  \begin{enumerate}
  \item[(a)] $\lpd(i,j)$ (along with the associated $\prob(i,j,.,.)$
    values) is a valid strategy dag, i.e., satisfies the conditions (i)
    and (ii) presented above.
  \item[(b)] The residual solution $\{z^j, w^j\}$ satisfies
    constraints~\eqref{eq:mabdaglp1}-\eqref{eq:mabdaglp3}.
  \item[(c)] For any time $t$ and state $u \in \Si$, $z^{j-1}_{u,t} -
    z^{j}_{u,t} = \prob(i,j,u,t)$.
  \end{enumerate}
\end{lemma}

\begin{proof}
We show the properties stated above one by one.

\noindent {\bf Property (a):} This follows from the construction of
\lref[Algorithm]{dalg:convex-sub}. More precisely, condition (i) is
satisfied because in \lref[Algorithm]{dalg:convex-sub} each $(u,t)$ is
visited at most once and that is the only time when a pair $(u,t)
\rightarrow (v, t')$ (with $t' \ge t + 1$) is added to the relation. For
condition (ii), notice that every time a pair $(u,t) \rightarrow (v,
t')$ is added to the relation we keep the invariant $\pprob(v, t') =
\sum_{(w,\tau)~\mathsf{s.t}~(i,j,w,\tau) \rightarrow (i,j,v,t')}
\pprob(w, \tau) \cdot p_{w,v}$; condition (ii) then follows since $\prob(.)$ is a
scaling of $\pprob(.)$.

\noindent {\bf Property (b):} Constraint~\eqref{eq:mabdaglp1} of
\ref{lp:mabdag} is clearly satisfied by the new LP solution $\{z^j,
w^j\}$ because of the two updates performed in
\lref[Steps]{dalg:convex5} and~\ref{dalg:convex6}: if we decrease the
$z$ value of any state at any time, the $w$ of all children are
appropriately reduced for the subsequent timestep.

Before showing that the solution $\{z^j, w^j\}$ satisfies
constraint~\eqref{eq:mabdaglp2}, we first argue that after every round
of the procedure they remain non-negative. By the choice of $\epsilon$
in \lref[step]{dalg:convex3}, we have $\prob(i,j,u,t) = \epsilon \cdot
\pprob(u,t) \leq \frac{z^{j-1}_{u,t}}{\pprob(u,t)}\pprob(u,t) =
z^{j-1}_{u,t}$ (notice that this inequality holds even if $\pprob(u,t) =
0$); consequently even after the update in \lref[step]{dalg:convex5},
$z^{j}_{u,t} \geq 0$ for all $u,t$. This and the fact that the
constraints~(\ref{eq:mabdaglp1}) are satisfied implies that $\{z^j,
w^j\}$ satisfies the non-negativity requirement.

We now show that constraint~\eqref{eq:mabdaglp2} is satisfied. Suppose
for the sake of contradiction there exist some $u \in \S$ and $t \in
[1,B]$ such that $\{z^j, w^j\}$ violates this constraint. Then, let us
consider any such $u$ and the earliest time $t_u$ such that the
constraint is violated.  For such a $u$, let $t'_u \leq t_u$ be the
latest time before $t_u$ where $z^{j-1}_{u,t'} > 0$. We now consider two
cases.

{\bf Case (i): $t'_u < t_u$}. This is the simpler case of the
two. Because $t_u$ was the earliest time where
constraint~\eqref{eq:mabdaglp2} was violated, we know that $\sum_{t'
  \leq t'_u} w^{j}_{u,t'} \geq \sum_{t' \leq t'_u}
z^{j}_{u,t'}$. Furthermore, since $z_{u,t}$ is never increased during
the course of the algorithm we know that $\sum_{t' = t'_u +1}^{t_u}
z^{j}_{u,t'} = 0$. This fact coupled with the non-negativity of
$w^j_{u,t}$ implies that the constraint in fact is not violated, which
contradicts our assumption about the tuple $u,t_u$.

{\bf Case (ii): $t'_u = t_u$}. In this case, observe that there cannot
be any pair of tuples $(v,t_1) \rightarrow (u,t_2)$ s.t.\ $t_1 < t_u$ and
$t_2 > t_u$, because any copy of $v$ (some ancestor of $u$) that is
played before $t_u$, will mark a copy of $u$ that occurs before $t_u$ or
the one being played at $t_u$ in \lref[Step]{dalg:peel3} of \peelStrat.
We will now show that summed over all $t' \leq t_u$, the decrease in the
LHS is counter-balanced by a corresponding drop in the RHS, between the
solutions $\{z^{j-1}, w^{j-1}\}$ and $\{z^{j}, w^{j}\}$ for this
constraint~\eqref{eq:mabdaglp2} corresponding to $u$ and $t_u$.  To this
end, notice that the only times when $w_{u,t'}$ is updated (in
\lref[Step]{dalg:convex6}) for $t' \leq t_u$, are when considering some
$(v,t_1)$ in \lref[Step]{dalg:convex3a} such that $(v,t_1) \rightarrow
(u,t_2)$ and $t_1 < t_2 \leq t_u$.  The value of $w_{u, t_1+1}$ is
dropped by exactly $\prob(i,j,v,t_1) \cdot p_{v,u}$. But notice that the
corresponding term $z_{u,t_2}$ drops by $\prob(i,j,u,t_2) =
\sum_{(v'',t'')~\mathsf{s.t}~(v'',t'')\rightarrow (u,t_2)}
\prob(i,j,v'',t'') \cdot p_{v'',u}$. Therefore, the total drop in $w$ is
balanced by a commensurate drop in $z$ on the RHS.

Finally, constraint~\eqref{eq:mabdaglp3} is also satisfied as the $z$
variables only decrease in value.

\noindent {\bf Property (c):} This is an immediate consequence of the
\lref[Step]{dalg:convex5} of the convex decomposition algorithm.
\end{proof}

As a consequence of the above lemma, we get the following.
\begin{lemma}
  \label{dlem:convexppt}
  Given a solution to~(\ref{lp:mabdag}), there exists a collection of
  at most $nB^2|\S|$ strategy dags $\{\lpd(i,j)\}$ such that $z_{u,t} =
  \sum_{j} \prob(i,j,u,t)$. Hence, $\sum_{(i, j, u)} \prob(i,j,u,t) \leq
  1$ for all $t$.
\end{lemma}

\subsection{Phases II and III}
\label{dsec:phase-iii}

We now show how to execute the strategy dags $\lpd(i,j)$. At a high
level, the development of the plays mirrors that of
\lref[Sections]{sec:phase-ii} and \ref{sec:phase-iii}. First we
transform $\lpd(i,j)$ into a (possibly exponentially large) blown-up
tree and show how this playing these exactly captures playing the
strategy dags. Hence (if running time is not a concern), we can simply
perform the gap-filling algorithm and make plays on these blown-up trees
following Phases II and III in \lref[Sections]{sec:phase-ii}
and~\ref{sec:phase-iii}. To achieve polynomial running time, we then
show that we can \emph{implicitly execute} the gap-filling phase while
playing this tree, thus getting rid of actually performing
\lref[Phase]{sec:phase-ii}. Finally, to complete our argument, we show
how we do not need to explicitly construct the blown-up tree, and can
generate the required portions depending on the transitions made thus
far \emph{on demand}.

\subsubsection{Transforming the DAG into a Tree}


Consider any strategy dag $\lpd(i,j)$. We first transform this dag into
a (possibly exponential) tree by making as many copies of a node $(i,j,u,t)$
as there are paths from the root to $(i,j,u,t)$ in $\lpd(i,j)$. More
formally, define $\lpdt(i,j)$ as the tree whose vertices are the simple
paths in $\lpd(i,j)$ which start at the root. To avoid confusion, we will explicitly refer to vertices of the tree
$\lpdt$ as tree-nodes, as distinguished from the
\emph{nodes} in $\lpd$; to simplify the notation we
identify each tree-node in $\lpdt$ with its corresponding path in
$\lpd$. Given two tree-nodes $P, P'$ in $\lpdt(i,j)$, add an arc from
$P$ to $P'$ if $P'$ is an immediate extension of $P$, i.e., if $P$
corresponds to some path $(i,j,u_1, t_1) \rightarrow \ldots \rightarrow
(i,j,u_k, t_k)$ in $\lpd(i,j)$, then $P'$ is a path $(i,j,u_1,t_1)
\rightarrow \ldots \rightarrow (i,j,u_k,t,k) \rightarrow
(i,j,u_{k+1},t_{k+1})$ for some node $(i,j,u_{k+1},t_{k+1})$.

For a tree-node $P \in \lpdt(i,j)$ which
corresponds to the path $(i,j,u_1, t_1) \rightarrow \ldots \rightarrow
(i,j,u_k,t_k)$ in $\lpd(i,j)$, we define $\state(P) = u_k$, i.e.,
$\state(\cdot)$ denotes the final state (in $\S_i$) in the path $P$. Now,
for tree-node $P \in \lpdt(i,j)$, if $u_1, \ldots, u_k$ are the children
of $\state(P)$ in $\Si$ with positive transition probability from $\state(P)$, then $P$ has exactly $k$ children $P_1, \ldots,
P_k$ with $\state(P_l)$ equal to $u_l$ for all $l \in [k]$. The
\emph{depth} of a tree-node $P$ is defined as the depth of $\state(P)$.
	
We now define the quantities $\ptime$ and $\prob$ for tree-nodes in
$\lpdt(i,j)$. Let $P$ be a path in $\lpd(i,j)$ from $\rho_i$ to node
$(i,j,u,t)$. We define $\ptime(P) := t$ and $\prob(P) :=
\prob(P')
  p_{(\state(P'),u)}$, where $P'$ is obtained by dropping the last node
  from $P$.  The blown-up tree $\lpdt(i,j)$ of our running example
$\lpd(i,j)$ (\lref[Figure]{dfig:dag}) is given in
\lref[Figure]{dfig:blown-up}.
\begin{lemma}
  For any state $u$ and time $t$, $\sum_{P~\mathsf{s.t}~\ptime(P) =
    t~\mathsf{and}~\state(P)=u} \prob(P) = \prob(i,j,u,t)$.
\end{lemma}

\begin{figure}[ht]
\centering
\includegraphics[scale=0.5]{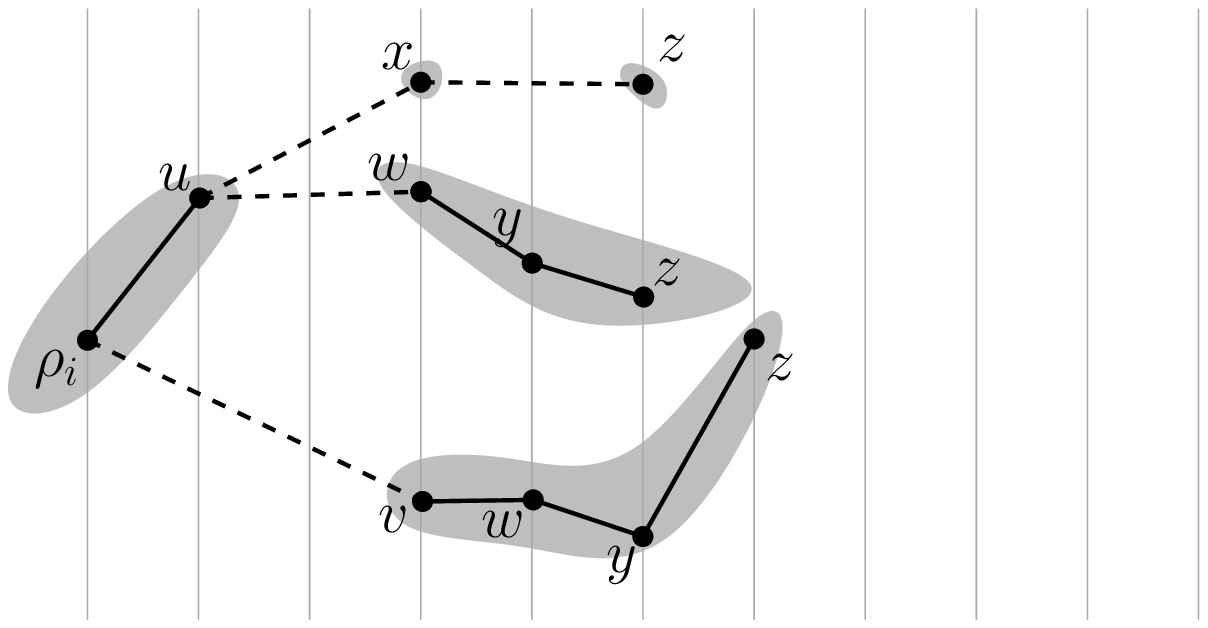}
\caption{Blown-up Strategy Forest $\lpdt(i,j)$}
\label{dfig:blown-up}
\end{figure}

Now that we have a tree labeled with $\prob$ and $\ptime$ values, the
notions of connected components and heads from
Section~\ref{sec:phase-ii} carry over. Specifically, we define
$\head(P)$ to be the ancestor $P'$ of $P$ in $\lpdt(i,j)$ with least
depth such that there is a path $(P' = P_1 \rightarrow \ldots
\rightarrow P_l = P)$ satisfying $\ptime(P_i) = \ptime(P_{i-1}) + 1$ for
all $i \in [2,l]$, i.e., the plays are made contiguously from $\head(P)$
to $P$ in the blown-up tree. We also define $\comp(P)$ as the set of all
tree-nodes $P'$ such that $\head(P) = \head(P')$.
	
In order to play the strategies $\lpdt(i,j)$ we first eliminate small
gaps. The algorithm \textsf{GapFill} presented in
\lref[Section]{sec:phase-ii} can be employed for this purpose and
returns trees $\lpdt'(i,j)$ which satisfy the analog of
\lref[Lemma]{lem:gapfill}.
	
\begin{lemma} \label{lem:gapfillDAG} The trees returned by
  \textsf{GapFill} satisfy the followings properties.
  \begin{OneLiners}
  \item[(i)] For each tree-node $P$ such that $r_{\state(P)} > 0$,
    $\ptime(\head(P)) \ge 2 \cdot \depth(\head(P))$.
  \item[(ii)] The total extent of plays at any time $t$, i.e.,
    $\sum_{P: \ptime(P)=t} \prob(P)$ is at most $3$.
  \end{OneLiners}
\end{lemma}

Now we use \lref[Algorithm]{alg:roundmab} to play the trees
$\lpdt(i,j)$. We restate the algorithm to conform with the notation used
in the trees $\lpdt(i,j)$.

\begin{algorithm}[ht!]
  \caption{Scheduling the Connected Components: Algorithm \textsf{AlgDAG}}
  \begin{algorithmic}[1]
  \label{alg:roundmabDAG}
    \STATE for arm $i$, \textbf{sample} strategy $\lpdt(i,j)$ with
    probability $\frac{\prob(\mroot(\lpdt(i,j)))}{24}$; ignore arm
    $i$ w.p.\ $1 - \sum_{j}
    \frac{\prob(\mroot(\lpdt(i,j)))}{24}$.  \label{alg:mabstep1DAG}
    \STATE let $A \gets$ set of ``active'' arms which chose a
    strategy in the random process. \label{alg:mabstep2DAG}
    \STATE for each $i \in A$, \textbf{let} $\sigma(i) \gets$ index $j$
    of the chosen $\lpdt(i,j)$ and \textbf{let} $\currnode(i) \gets $ root of $\lpdt(i,\sigma(i))$. \label{alg:mabstep3DAG}
    \WHILE{active arms $A \neq \emptyset$}
    \STATE \textbf{let} $i^* \gets$ arm with tree-node played earliest (i.e., $i^* \gets \argmin_{i \in A} \{ \ptime(\currnode(i))
    \}$).  \label{alg:mabstep4DAG}
    \STATE \textbf{let} $\tau \gets \ptime(\currnode(i^*))$.
		\WHILE{$\ptime(\currnode(i^*)) \neq \infty$
      \textbf{and} $\ptime(\currnode(i^*)) = \tau$} \label{alg:mabLoopDAG}
    \STATE \textbf{play} arm $i^*$ at state $\state(\currnode(i^*))$ \label{alg:mabPlayDAG}
    \STATE \textbf{let} $u$ be the new state of arm $i^*$ and \textbf{let} $P$ be the child of $\currnode(i^*)$ satisfying $\state(P) = u$.
    \STATE \textbf{update} $\currnode(i^*)$ to be $P$; \textbf{let} $\tau \gets \tau + 1$. \label{alg:mabstep5DAG}
    \ENDWHILE \label{alg:mabEndLoopDAG}
    \IF{$\ptime(\currnode(i^*)) = \infty$} \label{alg:mabAbandonDAG}
    \STATE \textbf{let} $A \gets A \setminus \{i^*\}$
    \ENDIF
    \ENDWHILE
 \end{algorithmic}
\end{algorithm}

Now an argument identical to that for Theorem~\ref{thm:main-mab} gives
us the following: 
\begin{theorem}
  \label{thm:main-mabDAG}
  The reward obtained by the algorithm~\textsf{AlgDAG} is at least a
  constant fraction of the optimum for \eqref{lp:mabdag}.
\end{theorem}

\subsubsection{Implicit gap filling}

Our next goal is to execute \textsf{GapFill} implicitly, that is, to
incorporate the gap-filling within Algorithm~\textsf{AlgDAG} without
having to explicitly perform the advances.

To do this, let us review some properties of the trees returned by
\textsf{GapFill}. For a tree-node $P$ in $\lpdt(i,j)$, let
$\ptime(P)$ denote the associated time in the original tree (i.e.,
before the application of \textsf{GapFill}) and let $\ptime'(P)$ denote
the time in the modified tree (i.e., after $\lpdt(i,j)$ is modified by
\textsf{GapFill}).

\begin{claim} \label{cl:gapppt}
For a non-root tree-node $P$ and its parent $P'$,
$\ptime'(P) = \ptime'(P') + 1$ if and only if, either $\ptime(P)
= \ptime(P') + 1$ or $2 \cdot \depth(P) > \ptime(P)$.
\end{claim}

\begin{proof}
Let us consider the forward direction. Suppose $\ptime'(P) = \ptime'(P') + 1$ but $\ptime(P) >
\ptime(P') + 1$. Then $P$ must have been the head of its component in the
original tree and an \textbf{advance} was performed on it, so we must
have $2 \cdot \depth(P) > \ptime(P)$.

For the reverse direction, if
$\ptime(P) = \ptime(P') + 1$ then $P$ could not have been a head since it belongs to the same component as $P'$ and hence it will always remain in the same component as $P'$ (as \textsf{GapFill} only merges components and never breaks them apart). Therefore, $\ptime'(P) = \ptime'(P') + 1$. On the other hand, if
$\ptime(P) > \ptime(P') + 1$ and $2 \cdot \depth(P) > \ptime(P)$, then
$P$ was a head in the original tree, and because of the above criterion, \textsf{GapFill} must have made an advance on $P'$ thereby including it in the same component as $P$; so again it is easy to see that $\ptime'(P) = \ptime'(P') + 1$.
\end{proof}

The crucial point here is that whether or not $P$ is in the same component
as its predecessor after the gap-filling (and, consequently, whether it was played contiguously along with its predecessor should that transition happen in~\textsf{AlgDAG})
can be inferred from the $\ptime$ values of $P, P'$ before gap-filling
and from the depth of $P$---it does not depend on any
other \textbf{advance}s that happen during the gap-filling.

Algorithm~\ref{alg:implicitFill} is a procedure
which plays the original trees $\lpdt(i,j)$ while implicitly performing
the \textbf{advance} steps of \textsf{GapFill} (by checking if the properties of Claim~\ref{cl:gapppt} hold). This change is reflected
in \lref[Step]{impalg:fill} where we may play a node even if it is not
contiguous, so long it satisfies the above stated properties.
Therefore, as a consequence of Claim~\ref{cl:gapppt}, we get the following Lemma that the plays made by \textsf{ImplicitFill} are identical to those made by \textsf{AlgDAG} after running \textsf{GapFill}.
	
\begin{algorithm}[ht!]
  \caption{Filling gaps implicitly: Algorithm \textsf{ImplicitFill}}
  \begin{algorithmic}[1]
  \label{alg:implicitFill}
    \STATE for arm $i$, \textbf{sample} strategy $\lpdt(i,j)$ with
    probability $\frac{\prob(\mroot(\lpdt(i,j)))}{24}$; ignore arm
    $i$ w.p.\ $1 - \sum_{j}
    \frac{\prob(\mroot(\lpdt(i,j)))}{24}$.  \label{impalg:mabstep1DAG}
    \STATE let $A \gets$ set of ``active'' arms which chose a
    strategy in the random process.
    \STATE  for each $i \in A$, \textbf{let} $\sigma(i) \gets$ index $j$
    of the chosen $\lpdt(i,j)$ and \textbf{let} $\currnode(i) \gets $ root of $\lpdt(i,\sigma(i))$. \label{impalg:rootchoose}
    \WHILE{active arms $A \neq \emptyset$}
    \STATE \textbf{let} $i^* \gets$ arm with state played earliest (i.e., $i^* \gets \argmin_{i \in A} \{ \ptime(\currnode(i))
    \}$).
    \STATE \textbf{let} $\tau \gets \ptime(\currnode(i^*))$.
		\WHILE{$\ptime(\currnode(i^*)) \neq \infty$
      \textbf{and} ($\ptime(\currnode(i^*)) = \tau$ \textbf{or} $2 \cdot  \depth(\currnode(i^*)) > \ptime(\currnode(i^*))$) } \label{impalg:fill}
    \STATE \textbf{play} arm $i^*$ at state $\state(\currnode(i^*))$
    \label{impalg:play}
    \STATE \textbf{let} $u$ be the new state of arm $i^*$ and \textbf{let} $P$ be the child of $\currnode(i^*)$ satisfying $\state(P) = u$. \label{impalg:nextNode}
    \STATE \textbf{update} $\currnode(i^*)$ to be $P$; \textbf{let} $\tau \gets \tau + 1$.
    \ENDWHILE
    \IF{$\ptime(\currnode(i^*)) = \infty$}
    \STATE \textbf{let} $A \gets A \setminus \{i^*\}$
    \ENDIF
    \ENDWHILE
 \end{algorithmic}
\end{algorithm}	

\begin{lemma}
  Algorithm $\textsf{ImplicitFill}$ obtains the same reward as algorithm
  $\textsf{AlgDAG}\circ\textsf{GapFill}$.
\end{lemma}

\subsubsection{Running \textbf{ImplicitFill} in Polynomial Time}
With the description of \textsf{ImplicitFill}, we are almost complete with our proof with the exception of handling the exponential blow-up incurred in moving from $\lpd$ to $\lpdt$. To resolve this, we now argue that while the blown-up $\lpdt$ made it easy to visualize the transitions and plays made, all of it can be done implicitly from the strategy DAG $\lpd$.
Recall that the tree-nodes in $\lpdt(i,j)$ correspond to simple paths in $\lpd(i,j)$. In the following, the final algorithm we employ (called \textsf{ImplicitPlay}) is simply the algorithm \textsf{ImplicitFill}, but with the exponentially blown-up trees $\lpdt(i, \sigma(i))$ being generated \emph{on-demand}, as the different transitions are made. We now describe how this can be done.

In Step~\ref{impalg:rootchoose} of \textsf{ImplicitFill},
we start off at the roots of the trees $\lpdt(i,\sigma(i))$, which
corresponds to the single-node path corresponding to the root of
$\lpd(i,\sigma(i))$. Now, at some point in time in the execution of \textsf{ImplicitFill}, suppose we are at the
tree-node $\currnode(i^*)$, which corresponds to a path $Q$ in
$\lpd(i,\sigma(i))$ that ends at $(i, \sigma(i), v, t)$ for some
$v$ and $t$. The invariant we maintain is that, in our algorithm \textsf{ImplicitPlay}, we are at node $(i, \sigma(i), v, t)$ in $\lpd(i,\sigma(i))$. Establishing this invariant would show that the two runs \textsf{ImplicitPlay} and \textsf{ImplicitFill} would be identical, which when coupled with Theorem~\ref{thm:main-mabDAG} would complete the proof---the information that \textsf{ImplicitFill} uses of $Q$, namely $\ptime(Q)$ and $\depth(Q)$, can be obtained from $(i, \sigma(i), v, t)$.

The invariant is clearly satisfied at the beginning, for the different root nodes.
Suppose it is true for some tree-node $\currnode(i)$, which corresponds to a path $Q$ in
$\lpd(i,\sigma(i))$ that ends at $(i, \sigma(i), v, t)$ for some
$v$ and $t$.
Now, suppose upon playing the arm $i$ at state $v$ (in Step~\ref{impalg:play}), we make a transition to state $u$ (say),
then \textsf{ImplicitFill} would find the unique child tree-node $P$ of
$Q$ in $\lpdt(i, \sigma(i))$ with $\state(P) = u$. Then let $(i, \sigma(i), u, t')$ be the last node of the path $P$, so that $P$ equals $Q$ followed by $(i, \sigma(i), u, t')$.

But, since the tree $\lpdt(i, \sigma(i))$ is
just an expansion of $\lpd(i, \sigma(i))$, the unique child $P$ in $\lpdt(i,\sigma(i))$ of tree-node $Q$ which has $\state(P) = u$, is (by definition of $\lpdt$) the unique node $(i, \sigma(i), u, t')$ of $\lpd(i, \sigma(i))$ such that $(i, \sigma(i), v, t) \rightarrow (i,\sigma(i), u,t')$.
Hence, just as \textsf{ImplicitFill} transitions
to $P$ in $\lpdt(i, \sigma(i))$ (in Step~\ref{impalg:nextNode}), we can transition to the state
$(i, \sigma(i), u, t')$ with just $\lpd$ at our disposal, thus establishing the invariant.

For completeness, we present the implicit algorithm below.

\begin{algorithm}[ht!]
  \caption{Algorithm \textsf{ImplicitPlay}}
  \begin{algorithmic}[1]
  \label{alg:implicitPlay}
    \STATE for arm $i$, \textbf{sample} strategy $\lpd(i,j)$ with
    probability $\frac{\prob(\mroot(\lpd(i,j)))}{24}$; ignore arm
    $i$ w.p.\ $1 - \sum_{j}
    \frac{\prob(\mroot(\lpd(i,j)))}{24}$.  
    \STATE let $A \gets$ set of ``active'' arms which chose a
    strategy in the random process.
    \STATE  for each $i \in A$, \textbf{let} $\sigma(i) \gets$ index $j$
    of the chosen $\lpd(i,j)$ and \textbf{let} $\currnode(i) \gets $ root of $\lpd(i,\sigma(i))$. 
    \WHILE{active arms $A \neq \emptyset$}
    \STATE \textbf{let} $i^* \gets$ arm with state played earliest (i.e., $i^* \gets \argmin_{i \in A} \{ \ptime(\currnode(i))
    \}$).
    \STATE \textbf{let} $\tau \gets \ptime(\currnode(i^*))$.
		\WHILE{$\ptime(\currnode(i^*)) \neq \infty$
      \textbf{and} ($\ptime(\currnode(i^*)) = \tau$ \textbf{or} $2 \cdot  \depth(\currnode(i^*)) > \ptime(\currnode(i^*))$) } 
    \STATE \textbf{play} arm $i^*$ at state $\state(\currnode(i^*))$
    \STATE \textbf{let} $u$ be the new state of arm $i^*$. 
    \STATE \textbf{update} $\currnode(i^*)$ to be $u$; \textbf{let} $\tau \gets \tau + 1$.
    \ENDWHILE
    \IF{$\ptime(\currnode(i^*)) = \infty$}
    \STATE \textbf{let} $A \gets A \setminus \{i^*\}$
    \ENDIF
    \ENDWHILE
 \end{algorithmic}
\end{algorithm}

%% file: bad-egs.tex
\section{Some Bad Examples}
\label{sec:egs}

\subsection{Badness Due to Cancelations}
\label{sec:badness-cancel}

We first observe that the LP relaxation for the \sks problem used
in~\cite{DeanGV08} has a large integrality gap in the model where cancelations are allowed,
\emph{even when the rewards are fixed for any item}. This was also noted
in~\cite{Dean-thesis}. Consider the following example: there are $n$
items, every item instantiates to a size of $1$ with probability $0.5$
or a size of $n/2$ with probability $0.5$, and its reward is always $1$.
Let the total size of the knapsack be $B = n$.  For such an instance, a
good solution would cancel any item that does not terminate at
size $1$; this way, it can collect a reward of at least $n/2$ in
expectation, because an average of $n/2$ items will instantiate with a
size $1$ and these will all contribute to the reward. On the other hand,
the LP from~\cite{DeanGV08} has value $O(1)$, since the mean size of any
item is at least $n/4$. In fact, any strategy that does not cancel jobs will
also accrue only $O(1)$ reward.

\subsection{Badness Due to Correlated Rewards}
\label{sec:badness-corr}
While the LP relaxations used for \mab (e.g., the formulation in ~\cite{GuhaM-stoc07})  can handle the issue explained above w.r.t cancelations,
      we now present an example of stochastic knapsack (where the reward is correlated with the actual size)
	for which the existing \mab LP formulations all have a large integrality gap.

	Consider the following example: there are $n$
items, every item instantiates to a size of $1$ with probability $1-1/n$
or a size of $n$ with probability $1/n$, and its reward is $1$ only if its size is $n$, and $0$ otherwise.
Let the total size of the knapsack be $B = n$.
Clearly, any integral solution can fetch an expected reward of $1/n$ --- if the first item it schedules instantiates to a large size, then it gives us a reward. Otherwise, no subsequent item can be fit within our budget even if it instantiates to its large size.
The issue with the existing LPs is that the \emph{arm-pull} constraints are ensured locally, and there is one global budget.
That is, even if we play each arm to completion individually, the expected size (i.e., number of pulls) they occupy is $1 \cdot (1-1/n) + n \cdot (1/n) \leq 2$. Therefore, such LPs can accommodate $n/2$ jobs, fetching a total reward of $\Omega(1)$. This example brings to attention the fact that all these item are competing to be pulled in the first time slot (if we begin an item in any later time slot it fetches zero reward), thus naturally motivating our time-indexed LP formulation in Section~\ref{sec:large}.

In fact, the above example also shows that if we allow ourselves a budget of $2B$, i.e., $2n$ in this case, we can in fact achieve an expected reward of $O(1)$ (much higher than what is possible with a budget of $B$) --- keep playing all items one by one, until one of them does not step after size $1$ and then play that to completion; this event happens with probability $\Omega(1)$.

\subsection{Badness Due to the Non-Martingale Property in MAB: The Benefit of Preemption}
\label{sec:preemption-gap}
Not only do cancelations help in our problems (as can be seen from the example in Appendix~\ref{sec:badness-cancel}), we now show that even \emph{preemption} is necessary in the case of \mab where the rewards do not satisfy the martingale property. In fact, this brings forward another key difference between our rounding scheme and earlier algorithms for \mab --- the necessity of preempting arms is not an artifact of our algorithm/analysis but, rather, is unavoidable.

Consider the following instance. There are $n$ identical arms, each of them with the following (recursively defined) transition tree starting at $\rho(0)$:

When the root $\rho(j)$ is pulled for $j < m$, the following two transitions can happen:
\begin{enumerate}
\item[(i)] with probability $1/(n \cdot n^{m-j})$, the arm transitions to the ``right-side'', where if it makes $B - n(\sum_{k=0}^{j} L^k)$ plays, it will deterministically reach a state with reward $n^{m-j}$. All intermediate states have $0$ reward.
\item[(ii)] with probability $1 - 1/(n \cdot n^{m-j})$, the arm transitions to the ``left-side'', where if it makes $L^{j+1} - 1$ plays, it will deterministically reach the state $\rho(j+1)$. No state along this path fetches any reward.
\end{enumerate}

Finally, node $\rho(m)$ makes the following transitions when played: (i) with probability $1/n$, to a leaf state that has a reward of $1$ and the arm ends there; (ii) with probability $1-1/n$, to a leaf state with reward of $0$.

For the following calculations, assume that $B \gg L > n$ and $m \gg 0$.

\noindent {\bf Preempting Solutions.}
We first exhibit a preempting solution with expected reward $\Omega(m)$. The strategy plays $\rho(0)$ of all the arms until one of them transitions to the ``right-side'', in which case it continues to play this until it fetches a reward of $n^m$. Notice that any root which transitioned to the right-side can be played to completion, because the number of pulls we have used thus far is at most $n$ (only those at the $\rho(0)$ nodes for each arm), and the size of the right-side is exactly $B - n$.
Now, if all the arms transitioned to the left-side, then it plays the $\rho(1)$ of each arm until one of them transitioned to the right-side, in which case it continues playing this arm and gets a reward of $n^{m-1}$.
Again, any root $\rho(1)$ which transitioned to the right-side \emph{can be played} to completion, because the number of pulls we have used thus far is at most $n(1 + L)$ (for each arm, we have pulled the root $\rho(0)$, transitioned the walk of length $L-1$ to $\rho(1)$ and then pulled $\rho(1)$), and the size of the right-side is exactly $B - n(1+L)$. This strategy is similarly defined, recursively.

We now calculate the expected reward: if any of the roots $\rho(0)$ made a transition to the right-side, we get a reward of $n^m$. This happens with probability roughly $1/n^m$, giving us an expected reward of $1$ in this case.
If all the roots made the transition to the left-side, then at least one of the $\rho(1)$ states will make a transition to their right-side with probability $\approx 1/n^{m-1}$ in which case will will get reward of $n^{m-1}$, and so on.
Thus, summing over the first $m/2$ such rounds, our expected reward is at least
\[ \frac{1}{n^m} n^m  + \left(1- \frac{1}{n^m}\right) \frac{1}{n^{m-1}}  n^{m-1} + \left(1- \frac{1}{n^m}\right) \left(1- \frac{1}{n^{m-1}}\right)\frac{1}{n^{m-2}}  n^{m-2} + \ldots
\]
Each term above is $\Omega(1)$ giving us a total of $\Omega(m)$ expected reward.

\noindent {\bf Non-Preempting Solutions.}
Consider any non-preempting solution. Once it has played the first node of an arm and it has transitioned to the left-side, it has to irrevocably decide if it abandons this arm or continues playing. But if it has continued to play (and made the transition of $L-1$ steps), then it cannot get any reward from the right-side of $\rho(0)$ of any of the other arms, because $L > n$ and the right-side requires $B-n$ pulls before reaching a reward-state.
Likewise, if it has decided to move from $\rho(i)$ to $\rho(i+1)$ on any arm, it cannot get \emph{any} reward from the right-sides of $\rho(0), \rho(1), \ldots, \rho(i)$ on \emph{any} arm due to budget constraints. Indeed, for any $i \geq 1$, to have reached $\rho(i+1)$ on any particular arm, it must have utilized $(1 + L-1)  + (1 + L^2 -1) + \ldots + (1 + L^{i+1} -1)$ pulls in total, which exceeds $n(1 +L + L^2 + \ldots + L^{i})$ since $L > n$. Finally, notice that if the strategy has decided to move from $\rho(i)$ to $\rho(i + 1)$ on any arm, the maximum reward that it can obtain is $n^{m - i - 1}$, namely, the reward from the right-side transition of $\rho(i + 1)$.

Using these properties, we observe that an optimal non-preempting strategy proceeds in rounds as described next.

\medskip \noindent {\bf Strategy at round $i$.} Choose a set $N_i$ of $n_i$ available arms and play them as follows: pick one of these arms, play until reaching state $\rho(i)$ and then play once more. If there is a right-side transition before reaching state $\rho(i)$, discard this arm since there is not enough budget to play until reaching a state with positive reward. If there is a right-side transition at state $\rho(i)$, play this arm until it gives reward of $n^{m - i}$. If there is no right-side transition and there is another arm in $N_i$ which is still to be played, discard the current arm and pick the next arm in $N_i$.

\medskip

In round $i$, at least $\max(0, n_i - 1)$ arms are discarded, hence $\sum_i n_i \le 2 n$. Therefore, the expected reward can be at most
\[ \frac{n_1}{n \cdot n^m} n^m  + \frac{n_2}{n \cdot n^{m-1}}  n^{m-1}  + \ldots + \frac{n_m}{n}  \leq 2
\]


%% file: app-knap.tex
\section{Proofs from Section~\ref{sec:nopmtn}}

\subsection{Proof of Theorem~\ref{thm:large}} \label{app:nopmtn-proof}
	Let $\mathsf{add}_i$ denote the event that item $i$ was added to the knapsack in \lref[Step]{alg:big3}. Also, let $V_i$ denote the random variable
corresponding to the reward that our algorithm gets from item
$i$.
		
Clearly if item $i$ has $D_i
= t$ and was added, then it is added to the knapsack before time $t$.
In this case it is easy to see that $\E[V_i \mid \mathsf{add}_i \wedge (D_i =
t)] \ge R_{i,t}$ (because its random size is independent of when the algorithm started it). Moreover, from the previous lemma we have that
$\Pr(\mathsf{add}_i \mid (D_i = t)) \ge 1/2$ and from
\lref[Step]{alg:big1} we have $\Pr(D_i = t) =
\frac{x^*_{i,t}}{4}$; hence $\Pr(\mathsf{add}_i
\wedge (D_i = t)) \ge x^*_{i,t}/8$. Finally adding over all
possibilities of $t$, we lower bound the expected value of
$V_i$ by $$\E[V_i] \ge \sum_t \E[V_i \mid
\mathsf{add}_i \wedge (D_i = t)] \cdot \Pr(\mathsf{add}_i
\wedge (D_i = t)) \ge \frac{1}{8} {\sum_t x^*_{i,t} R_{i,t}}.$$
		
Finally, linearity of expectation over all items shows that the total expected reward
of our algorithm is at least $\frac18 \cdot \sum_{i, t} x^*_{i,t}
R_{i,t} = \LPOpt/8$, thus completing the proof.

\subsection{Making \sksnocancel Fully Polynomial} \label{app:polytime-nopmtn}
	
Recall that our LP relaxation \ref{lp:large} in Section~\ref{sec:nopmtn}
uses a global time-indexed LP. In order to make it compact, our approach will be to group the $B$ timeslots
in \ref{lp:large} and show that the grouped LP has optimal value
within constant factor of \ref{lp:large}; furthermore, we show also that it can be rounded and analyzed
almost identically to the original LP. To this end, consider the following LP relaxation:
\begin{alignat}{2} \tag{$\mathsf{PolyLP}_L$} \label{lp:largePoly}
  \max &\ts \sum_i \sum_{j = 0}^{\log B} \er_{i,2^{j + 1}} \cdot x_{i,2^j} &\\
  &\ts \sum_{j = 0}^{\log B} x_{i,2^j} \le 1 &\forall i \label{LPbig1Poly}\\
  &\ts \sum_{i, j' \le j} x_{i,2^{j'}} \cdot \E[\min(S_i,2^{j+1})]
  \le 2\cdot 2^j  \qquad &\forall j \in [0, \log B] \label{LPbig2Poly}\\
  &x_{i,2^j} \in [0,1] &\forall j \in [0, \log B], \forall i
\end{alignat}	
	
The next two lemmas relate the value of \eqref{lp:largePoly} to that of
the original LP \eqref{lp:large}.
	
\begin{lemma} \label{lemma:largePoly1} The optimum of
  \eqref{lp:largePoly} is at least half of the optimum of
  \eqref{lp:large}.
\end{lemma}
	
\begin{proof}
  Consider a solution $x$ for \eqref{lp:large} and define $\bar{x}_{i1}
  = x_{i,1}/2 + \sum_{t \in [2,4)} x_{i,t}/2$ and $\bar{x}_{i,2^j} =
  \sum_{t \in [2^{j + 1}, 2^{j + 2})} x_{i,t}/2$ for $1 < j \le \log
  B$. It suffices to show that $\bar{x}$ is a feasible solution to
  \eqref{lp:largePoly} with value greater than of equal to half of the
  value of $x$.
		
  For constraints \eqref{LPbig1Poly} we have $\sum_{j = 0}^{\log B}
  \bar{x}_{i,2^j} = \sum_{t \ge 1} x_{i,t}/2 \le 1/2$; these
  constraints are therefore easily satisfied. We now show that $\{\bar{x}\}$ also satisfies
  constraints \eqref{LPbig2Poly}:
  \begin{align*}
    &\sum_{i, j' \le j} x_{i,2^{j'}} \cdot \E[\min(S_i,2^{j+1})] = \sum_i \sum_{t = 1}^{2^{j+2} - 1} \frac{x_{i,t} \E[\min(S_i, 2^{j + 1})]}{2} \\
    & \le  \sum_i \sum_{t = 1}^{2^{j+2} - 1} \frac{x_{i,t} \E[\min(S_i, 2^{j + 2} - 1)]}{2} \le 2^{j + 2} - 1,
  \end{align*}
  where the last inequality follows from feasibility of $\{x\}$.
		
  Finally, noticing that $\er_{i,t}$ is non-increasing with respect to
  $t$, it is easy to see that $\sum_i \sum_{j = 0}^{\log B} \er_{i,2^{j
      + 1}} \cdot \bar{x}_{i,2^j} \ge \sum_{i,t} \er{i,t} \cdot
  x_{i,t}/2$ and hence $\bar{x}$ has value greater than of equal to half
  of the value of $x$ ad desired.
\end{proof}

\begin{lemma} \label{lemma:largePoly2} Let $\{\bar{x}\}$ be a feasible
  solution for \eqref{lp:largePoly}. Define $\{\hat{x}\}$ satisfying $\hat{x}_{i,t} =
  \bar{x}_{i,2^j}/2^j$ for all $t \in [2^j, 2^{j+1})$ and $i \in
  [n]$. Then $\{\hat{x}\}$ is feasible for \eqref{lp:large} and has value at least
  as large as $\{\bar{x}\}$.
\end{lemma}
	
\begin{proof}
  The feasibility of $\{\bar{x}\}$ directly imply that $\{\hat{x}\}$ satisfies
  constraints \eqref{LPbig1}. For constraints \eqref{LPbig2}, consider
  $t \in [2^j, 2^{j + 1})$; then we have the following:
  \begin{align*}
    & \sum_{i, t' \le t} \hat{x}_{i,t'} \cdot \E[\min(S_i,t)] \le \sum_i \sum_{j' \le j} \sum_{t \in [2^{j'}, 2^{j' + 1})} \frac{\bar{x}_{i,2^j}}{2^j} \E[\min(S_i, 2^{j+1})] \\
    & = \sum_i \sum_{j' \le j} \bar{x}_{i,2^j} \E[\min(S_i, 2^{j+1})] \le 2 \cdot 2^j \le 2t.
  \end{align*}

  Finally, again using the fact that $\er_{i,t}$ is non-increasing in
  $t$ we get that the value of $\{\hat{x}\}$ is
  \begin{align*}
    \sum_{i,t} \er_{i,t} \cdot \hat{x}_{i,t} = \sum_i \sum_{j = 0}^{\log B} \sum_{t \in [2^j, 2^{j+1})} \er_{i,t} \frac{\bar{x}_{i,2^j}}{2^j} \ge \sum_i \sum_{j = 0}^{\log B} \sum_{t \in [2^j, 2^{j+1})} \er_{i,2^{j + 1}} \frac{\bar{x}_{i,2^j}}{2^j} = \sum_i \sum_{j = 0}^{\log B} \er_{i,2^{j + 1}} \bar{x}_{i,2^j},
  \end{align*}
  which is then at least as large as the value of $\{\bar{x}\}$. This
  concludes the proof of the lemma.
\end{proof}

The above two lemmas show that the \ref{lp:largePoly} has value
close to that of \ref{lp:large}: let's now show that we can simulate
the execution of Algorithm \skslarge just given an optimal solution
$\{\bar{x}\}$ for \eqref{lp:largePoly}. Let $\{\hat{x}\}$ be defined as in the above
lemma, and consider the Algorithm \skslarge applied to $\{\hat{x}\}$. By the definition of $\{\hat{x}\}$, here's how to execute
\lref[Step]{alg:big1} (and hence the whole algorithm) in polynomial
time: we obtain $D_i = t$ by picking $j \in [0, \log B]$ with
probability $\bar{x}_{i,2^j}$ and then selecting $t \in [2^j, 2^{j +
  1})$ uniformly; notice that indeed $D_i = t$ (with $t \in [2^j, 2^{j +
  1})$) with probability $\bar{x}_{i,2^j}/2^j = \hat{x}_{i,t}$.
	
Using this observation we can obtain a $1/16$ approximation for our
instance $\mathcal{I}$ in polynomial time by finding the optimal
solution $\{\bar{x}\}$ for \eqref{lp:largePoly} and then running Algorithm
\skslarge over $\{\hat{x}\}$ as described in the previous paragraph. Using a
direct modification of \lref[Theorem]{thm:large} we have that the
strategy obtained has expected reward at least at large as $1/8$ of the
value of $\{\hat{x}\}$, which by \lref[Lemmas]{lemma:largePoly1} and
\ref{lemma:largePoly2} (and \lref[Lemma]{thm:lp-large-valid}) is
within a factor of $1/16$ of the optimal solution for $\mathcal{I}$.

\section{Proofs from Section~\ref{sec:sk}} \label{app:small}

\subsection{Proof of Lemma~\ref{lem:stop-dist}}
  The proof works by induction. For the base case, consider $t=0$.
  Clearly, this item is forcefully canceled in \lref[step]{alg:st:2} of Algorithm~\ref{alg:skssmall} \skssmall
  (in the iteration with $t=0$) with probability $s^*_{i,0}/v^*_{i,0} -
  \pi_{i,0}/\sum_{t' \geq 0} \pi_{i,t'}$. But since $\pi_{i,0}$ was
  assumed to be $0$ and $v^*_{i,0}$ is $1$, this quantity is exactly
  $s^*_{i,0}$, and this proves property~(i).  For property~(ii), item
  $i$ is processed for its $\mathbf{1}^{st}$ timestep if it did not get
  forcefully canceled in \lref[step]{alg:st:2}. This therefore happens
  with probability $1 - s^*_{i,0} = v^*_{i,0} - s^*_{i,0} = v^*_{i,1}$.
  For property~(iii), conditioned on the fact that it has been processed
  for its $\mathbf{1}^{st}$ timestep, clearly the probability that its (unknown)
  size has instantiated to $1$ is exactly $\pi_{i,1}/\sum_{t' \geq 1}
  \pi_{i,t'}$. When this happens, the job stops in \lref[step]{alg:st:5}, thereby establishing the
  base case.

  Assuming this property holds for every timestep until some fixed value
  $t-1$, we show that it holds for $t$; the proofs are very similar to
  the base case.  Assume item $i$ was processed for the $t^{th}$
  timestep (this happens w.p $v^*_{i,t}$ from property
  (ii) of the induction hypothesis). Then from property (iii), the
  probability that this item completes at this timestep is exactly
  $\pi_{i,t}/\sum_{t' \geq t} \pi_{i,t'}$. Furthermore, it gets
  forcefully canceled in \lref[step]{alg:st:2} with probability
  $s^*_{i,t}/v^*_{i,t} - \pi_{i,t}/\sum_{t' \geq t} \pi_{i,t'}$.  Thus
  the total probability of stopping at time $t$, assuming it has been
  processed for its $t^{th}$ timestep is exactly $s^*_{i,t}/v^*_{i,t}$; unconditionally, the probability of stopping at time $t$ is
hence   $s^*_{i,t}$.

  Property~(ii) follows as a consequence of Property~(i), because the
  item is processed for its $(t+1)^{st}$ timestep only if it did not
  stop at timestep $t$. Therefore, conditioned on being processed for
  the $t^{th}$ timestep, it continues to be processed with probability
  $1 - s^*_{i,t}/v^*_{i,t}$. Therefore, removing the conditioning, we
  get the probability of processing the item for its $(t+1)^{st}$
  timestep is $v^*_{i,t} - s^*_{i,t} = v^*_{i,t+1}$.  Finally, for
  property~(iii), conditioned on the fact that it has been processed for
  its $(t+1)^{st}$ timestep, clearly the probability that its (unknown)
  size has instantiated to exactly $(t+1)$ is $\pi_{i,t+1}/\sum_{t' \geq
    t+1} \pi_{i,t'}$. When this happens, the job stops in \lref[step]{alg:st:5} of the algorithm.

\subsection{\sks with Small Sizes: A Fully Polytime Algorithm} \label{app:polytime}
	
The idea is to quantize the possible sizes of the items in order to
ensure that LP \ref{lpone} has polynomial size, then obtain a good
strategy (via Algorithm \skssmall) for the transformed instance, and
finally to show that this strategy is actually almost as good for the
original instance.
	
Consider an instance $\mathcal{I} = (\pi, R)$ where $R_{i,t} = 0$ for
all $t > B/2$. Suppose we start scheduling an item at some time; instead
of making decisions of whether to continue or cancel an item at each
subsequent time step, we are going to do it in time steps which are
powers of 2. To make this formal, define instance $\bar{\mathcal{I}} =
(\bar{\pi}, \bar{R})$ as follows: set $\bar{\pi}_{i,2^j} = \sum_{t \in
  [2^j, 2^{j + 1})} \pi_{i,t}$ and $\bar{R}_{i,2^j} = (\sum_{t \in [2^j,
  2^{j+1})} \pi_{i,t} R_{i,t} )/ \bar{\pi}_{i,2^j}$ for all $i \in [n]$
and $j \in \{0, 1, \ldots, \lfloor\log B \rfloor\}$. The instances are
coupled in the natural way: the size of item $i$ in the instance
$\bar{\mathcal{I}}$ is $2^j$ iff the size of item $i$ in the instance
$\mathcal{I}$ lies in the interval $[2^j, 2^{j+1})$.
	
In \lref[Section]{caseSmall}, a \emph{timestep} of an item has duration
of 1 time unit. However, due to the construction of $\bar{\mathcal{I}}$,
it is useful to consider that the $t^{th}$ time step of an item has
duration $2^t$; thus, an item can only complete at its $0^{th}$,
$1^{st}$, $2^{nd}$, etc. timesteps. With this in mind, we can write an
LP analogous to \eqref{lpone}:
	\begin{alignat}{2}
  \max &\ts \sum_{1 \leq j \leq \log(B/2)} \sum_{1 \leq i \leq n}  v_{i,2^j} \cdot \bar{R}_{i,2^j}  \frac{\bar{\pi}_{i,2^j}}{\sum_{j' \geq j} \pi_{i,2^{j'}}}   & &
  \tag{$\mathsf{PolyLP}_{S}$} \label{lpone.2} \\
  & v_{i,2^j} = s_{i,2^j} + v_{i,2^j+1} & \qquad & \forall \,
  j \in [0,\log B], \, i \in [n]   \label{eq:1.2} \\
  &s_{i,2^j} \geq  \frac{\bar{\pi}_{i,2^j}}{\sum_{j' \geq j} \bar{\pi}_{i,2^{j'}}} \cdot v_{i,2^j} & \qquad & \forall \, t \in
  [0,\log B], \, i \in [n] \label{eq:2.2} \\
  &\ts \sum_{i \in [n]} \sum_{j \in [0, \log B]} 2^j \cdot s_{i,2^j}  \leq B
    &   \label{eq:3.2}\\
  &v_{i,0} = 1 & \qquad & \forall \, i \label{eq:4.2} \\
  v_{i,2^j}, s_{i,2^j} &\in [0,1] & \qquad & \forall \, j \in [0,\log B], \, i \in
  [n] \label{eq:5.2}
\end{alignat}
Notice that this LP has size polynomial in the size of the instance
$\mathcal{I}$.
	
Consider the LP \eqref{lpone} with respect to the instance $\mathcal{I}$
and let $(v, s)$ be a feasible solution for it with objective value
$z$. Then define $(\bar{v}, \bar{s})$ as follows: $\bar{v}_{i,2^j} =
v_{i,2^j}$ and $\bar{s}_{i,2^j} = \sum_{t \in [2^j, 2^{j + 1})}
s_{i,j}$. It is easy to check that $(\bar{v}, \bar{s})$ is a feasible
solution for \eqref{lpone.2} with value at least $z$, where the latter
uses the fact that $v_{i,t}$ is non-increasing in $t$.
	%
Using \lref[Theorem]{thm:lp1-valid} it then follows that the optimum of
\eqref{lpone.2} with respect to $(\bar{\pi}, \bar{R})$ is at least as
large as the reward obtained by the optimal solution for the stochastic
knapsack instance $(\pi, R)$.
	
Let $(\bar{v}, \bar{s})$ denote an optimal solution of
\eqref{lpone.2}. Notice that with the redefined notion of timesteps we
can naturally apply Algorithm \skssmall to the LP solution $(\bar{v},
\bar{s})$. Moreover, \lref[Lemma]{lem:stop-dist} still holds in this
setting. Finally, modify Algorithm \skssmall by ignoring items with
probability $1 - 1/8 = 7/8$ (instead of $3/4$) in \lref[Step]{alg:st:1}
(we abuse notation slightly and shall refer to the modified algorithm
also as \skssmall) and notice that \lref[Lemma]{lem:stop-dist} still
holds.
	
Consider the strategy $\bar{\mathbb{S}}$ for $\bar{\mathcal{I}}$
obtained from Algorithm \skssmall. We can obtain a strategy $\mathbb{S}$
for $\mathcal{I}$ as follows: whenever $\mathbb{S}$ decides to process
item $i$ of $\bar{\mathcal{I}}$ for its $j$th timestep, we decide to
continue item $i$ of $\mathcal{I}$ while it has size from $2^j$ to $2^{j
  + 1} - 1$.
	
\begin{lemma}
  Strategy $\mathbb{S}$ is a $1/16$ approximation for $\mathcal{I}$.
\end{lemma}
	
\begin{proof}
  Consider an item $i$. Let $\bar{O}$ be the random variable denoting
  the total size occupied before strategy $\bar{\mathbb{S}}$ starts
  processing item $i$ and similarly let $O$ denote the total size
  occupied before strategy $\mathbb{S}$ starts processing item
  $i$. Since \lref[Lemma]{lem:stop-dist} still holds for the modified
  algorithm \skssmall, we can proceed as in \lref[Theorem]{thm:small}
  and obtain that $\E[\bar{O}] \le B/8$. Due to the definition of
  $\mathbb{S}$ we can see that $O \le 2 \bar{O}$ and hence $\E[O] \le
  B/4$. From Markov's inequality we obtain that $\Pr(O \ge B/2) \le
  1/2$. Noticing that $i$ is started by $\mathbb{S}$ with probability
  $1/8$ we get that the probability that $i$ is started and there is at
  least $B/2$ space left on the knapsack at this point is at least
  $1/16$. Finally, notice that in this case $\bar{\mathbb{S}}$ and
  $\mathbb{S}$ obtain the same expected value from item $i$, namely
  $\sum_j \bar{v}_{i,2^j} \cdot \bar{R}_{i,2^j}
  \frac{\bar{\pi}_{i,2^j}}{\sum_{j' \geq j} \pi_{i,2^{j'}}}$. Thus
  $\mathbb{S}$ get expected value at least that of the optimum of
  \eqref{lpone.2}, which is at least the value of the optimal solution
  for $\mathcal{I}$ as argued previously.
\end{proof}


%% file: app-polytime.tex


%% file: app-mab.tex
\section{Details from Section~\ref{sec:mab}}

\subsection{Details of Phase~I (from Section~\ref{sec:phase-i})}
\label{sec:details-phase-i}

We first begin with some notation that will be useful in the algorithm below. For any state $u \in \Si$ such that the path from $\rho_i$ to $u$ follows the states $u_1 = \rho_i, u_2, \ldots, u_k = u$, let $\pi_u = \Pi_{l=1}^{k-1} p_{u_i, u_{i+1}}$.

Fix an arm $i$, for which we will perform the decomposition. Let $\{z, w\}$ be a feasible solution to \ref{lp:mab} and set $z^0_{u,t} = z_{u,t}$  and $w^0_{u,t} = w_{u,t}$ for all $u \in \Si$, $t \in [B]$. We will gradually alter the fractional solution as we build the different forests. We note that in a particular iteration with index $j$, all $z^{j-1}, w^{j-1}$ values that are not updated in \lref[Steps]{alg:convex5} and~\ref{alg:convex6} are retained in $z^j, w^j$ respectively.
\begin{algorithm}[ht!]
\caption{Convex Decomposition of Arm $i$}
\begin{algorithmic}[1]
\label{alg:convex}
	\STATE {\bf set} ${\cal C}_i \leftarrow \emptyset$ and {\bf set loop index} $j \leftarrow 1$.
	\WHILE {$\exists$ a node $u \in \Si$ s.t $\sum_{t} z^{j-1}_{u,t} > 0$} \label{alg:convex1}
	\STATE {\bf initialize} a new tree $\lpt(i,j) = \emptyset$. \label{alg:convex1a}
	\STATE {\bf set} $A \leftarrow \{u \in \Si ~\textsf{s.t}~\sum_{t} z^{j-1}_{u,t} > 0\}$. \label{alg:convex1b}
		\STATE for all $u \in \Si$, {\bf set} $\ptime(i,j,u) \leftarrow \infty$, $\prob(i,j,u) \leftarrow 0$, and {\bf set} $\epsilon_u \leftarrow \infty$.
		\FOR{ every $u \in A$}
			\STATE {\bf update} $\ptime(i,j,u)$ to the smallest time $t$ s.t $z^{j-1}_{u,t} > 0$. \label{alg:convex2}
			\STATE {\bf update} $\epsilon_u = {z^{j-1}_{u,\ptime(i,j,u)}}/{\pi_{u}}$ \label{alg:convex2a}
		\ENDFOR
		\STATE {\bf let} $\epsilon = \min_{u} \epsilon_u$. \label{alg:convex3}
		\FOR{ every $u \in A$}
			\STATE {\bf set} $\prob(i,j,u) = \epsilon \cdot \pi_u$. \label{alg:convex4}
			\STATE {\bf update} $z^j_{u,\ptime(i,j,u)} = z^{j-1}_{u, \ptime(i,j,u)} - \prob(i,j,u)$. \label{alg:convex5}
			\STATE {\bf update} $w^j_{v, \ptime(i,j,u)+1} = w^{j-1}_{v, \ptime(i,j,u)+1} - \prob(i,j,u) \cdot p_{u,v}$ for all $v$ s.t $\parent(v) = u$. \label{alg:convex6}
		\ENDFOR	
	\STATE {\bf set} ${\cal C}_i \leftarrow {\cal C}_i \cup \lpt(i,j)$.  	 \label{alg:convex7}
	\STATE {\bf increment} $j \leftarrow  j + 1$.
	\ENDWHILE
\end{algorithmic}
\end{algorithm}
For brevity of notation, we shall use ``iteration $j$ of \lref[step]{alg:convex1}'' to denote the execution of the entire block (\lref[steps]{alg:convex1a} -- \ref{alg:convex7}) which constructs strategy  forest $\lpt(i,j)$.

%

\begin{lemma} \label{lem:convexstep}
Consider an integer $j$ and suppose that $\{z^{j-1}, w^{j-1}\}$ satisfies constraints~\eqref{eq:mablp1}-\eqref{eq:mablp3} of \ref{lp:mab}. Then after iteration $j$ of \lref[Step]{alg:convex1}, the following properties hold:
\begin{enumerate}
\item[(a)] $\lpt(i,j)$ (along with the associated $\prob(i,j,.)$ and $\ptime(i,j,.)$ values) is a valid strategy forest, i.e., satisfies the conditions (i) and (ii) presented in Section \ref{sec:phase-i}.
\item[(b)] The residual solution $\{z^j, w^j\}$ satisfies constraints~\eqref{eq:mablp1}-\eqref{eq:mablp3}.
\item[(c)] For any time $t$ and state $u \in \Si$, $z^{j-1}_{u,t} - z^{j}_{u,t} = \prob(i,j,u) \mathbf{1}_{\ptime(i,j,u)=t}$.
\end{enumerate}
\end{lemma}

\begin{proof}
We show the properties stated above one by one.

\noindent {\bf Property (a):}
We first show that the $\ptime$ values satisfy $\ptime(i,j,u)$ $\geq$ $\ptime(i,j,\parent(u)) + 1$, i.e. condition (i) of strategy forests.
For sake of contradiction, assume that there exists $u \in \Si$ with $v = \parent(u)$ where $\ptime(i,j,u) \leq \ptime(i,j,v)$. Define $t_u = \ptime(i,j,u)$ and $t_v = \ptime(i,j,\parent(u))$; the way we updated $\ptime(i,j,u)$ in \lref[step]{alg:convex2} gives that $z^{j-1}_{u,t_u} > 0$.

Then, constraint~(\ref{eq:mablp2}) of the LP implies that $\sum_{t' \leq t_u} w^{j-1}_{u,t'} > 0$. In particular, there exists a time $t' \leq t_u \leq t_v$ such that $w^{j-1}_{u,t'} > 0$. But now, constraint~(\ref{eq:mablp1}) enforces that $z^{j-1}_{v, t'-1} = w^{j-1}_{u,t'}/p_{v, u} > 0$ as well. But this contradicts the fact that $t_v$ was the first time s.t $z^{j-1}_{v,t} > 0$. Hence we have $\ptime(i,j,u) \geq \ptime(i,j,\parent(u))+1$.

As for condition (ii) about $\prob(i,j,.)$, notice that if $\ptime(i,j,u) \neq \infty$, then $\prob(i,j,u)$ is set to $\epsilon \cdot \pi_u$ in \lref[step]{alg:convex4}. It is now easy to see from the definition of $\pi_u$ (and from the fact that $\ptime(i,j,u) \neq \infty \Rightarrow \ptime(i,j,\parent(u)) \neq \infty$) that $\prob(i,j,u) = \prob(i,j,\parent(u)) \cdot p_{\parent(u),u}$.

\noindent {\bf Property (b):} Constraint~\eqref{eq:mablp1} of \ref{lp:mab} is clearly satisfied by the new LP solution $\{z^j, w^j\}$ because of the two updates performed in \lref[Steps]{alg:convex5} and~\ref{alg:convex6}: if we decrease the $z$ value of any node at any time, the $w$ of all children are appropriately reduced (for the subsequent timestep).

Before showing that the solution $\{z^j, w^j\}$ satisfies constraint~\eqref{eq:mablp2}, we first argue that they remain non-negative. By the choice of $\epsilon$ in step~\ref{alg:convex3}, we have $\prob(i,j,u) = \epsilon  \pi_u \leq \epsilon_u \pi_u = z^{j-1}_{u,\ptime(i,j,u)}$ (where $\epsilon_u$ was computed in \lref[Step]{alg:convex2a}); consequently even after the update in step~\ref{alg:convex5}, $z^{j}_{u,\ptime(i,j,u)} \geq 0$ for all $u$. This and the fact that the constraints~(\ref{eq:mablp1}) are satisfied implies that $\{z^j, w^j\}$ satisfies the non-negativity requirement.

We now show that constraint~\eqref{eq:mablp2} is satisfied.
For any time $t$ and state $u \notin A$ (where $A$ is the set computed in step~\ref{alg:convex1b} for iteration $j$), clearly it must be that $\sum_{t' \leq t} z^{j-1}_{u,t} = 0$ by definition of the set $A$; hence just the non-negativity of $w^j$ implies that these constraints are trivially satisfied.

Therefore consider some $t \in [B]$ and a state $u \in A$. We know from step~\ref{alg:convex2} that $\ptime(i,j,u) \neq \infty$. If $t < \ptime(i,j,u)$, then the way $\ptime(i,j,u)$ is updated in step~\ref{alg:convex2} implies that $\sum_{t' \leq t} z^j_{u,t'} = \sum_{t' \leq t} z^{j-1}_{u,t'} = 0$, so the constraint is trivially satisfied because $w^j$ is non-negative. If $t \geq \ptime(i,j,u)$, we claim that the change in the left hand side and right hand side (between the solutions $\{z^{j-1}, w^{j-1}\}$ and $\{z^{j}, w^{j}\}$) of the constraint under consideration is the same, implying that it will be still satisfied by $\{z^j, w^j\}$.

To prove this claim, observe that the right hand side has decreased by exactly $z^{j-1}_{u, \ptime(i,j,u)} - z^{j}_{u, \ptime(i,j,u)}  = \prob(i,j,u)$.
But the only value which has been modified in the left hand side is $w^{j-1}_{u, \ptime(i,j,\parent(u))+1}$, which has gone down by $\prob(i, j, \parent(u)) \cdot p_{\parent(u), u}$. Because $\lpt(i,j)$ forms a valid strategy forest, we have $\prob(i,j,u) = \prob(i,j,\parent(u)) \cdot p_{\parent(u), u}$, and thus the claim follows.

Finally, constraint~\eqref{eq:mablp3} are also satisfied as the $z$ variables only decrease in value over iterations.

\noindent {\bf Property (c):} This is an immediate consequence of the \lref[Step]{alg:convex5}.
\end{proof}

To prove \lref[Lemma]{lem:convexppt},  firstly notice that since $\{z^0, w^0\}$ satisfies constraints~\eqref{eq:mablp1}-\eqref{eq:mablp3}, we can proceed by induction and infer that the properties in the previous lemma hold for every strategy forest in the decomposition; in particular, each of them is a valid strategy forest.

In order to show that the marginals are preserved,
observe that in the last iteration $j^*$ of procedure we have $z^{j^*}_{u,t} = 0$ for all $u, t$. Therefore, adding the last property in the previous lemma over all $j$ gives
\begin{align*}
	z_{u,t} = \sum_{j \ge 1} (z^{j - 1}_{u,t} - z^j_{u,t}) = \sum_{j \ge 1} \prob(i,j,u) \mathbf{1}_{\ptime(i,j,u) = t} = \sum_{j : \ptime(i,j,u) = t} \prob(i,j,u).
\end{align*}

Finally, since some $z^j_{u,t}$ gets altered to $0$ since in each iteration of the above algorithm, the number of strategies for each arm in the decomposition is upper bounded by $B|\S|$. This completes the proof of
\lref[Lemma]{lem:convexppt}.


\subsection{Details of Phase~II (from Section~\ref{sec:phase-ii})}
\label{sec:details-phase-ii}

\begin{proofof}{Lemma~\ref{lem:gapfill}}
Let $\ptime^t(u)$ denote the time assigned to node $u$ by the end of round $\tau = t$ of the algorithm; $\ptime^{B + 1}(u)$ is the initial time of $u$. Since the algorithm works backwards in time, our round index will start at $B$ and end up at $1$.
	To prove property (i) of the statement of the lemma, notice that the algorithm only converts head nodes to non-head nodes and not the other way around. Moreover, heads which survive the algorithm have the same $\ptime$ as originally. So it suffices to show that heads which originally did not satisfy property (i)---namely, those with $\ptime^{B + 1}(v) < 2 \cdot \depth(v)$---do not survive the algorithm; but this is clear from the definition of Step \ref{alg:gap1}.

	To prove property (ii), fix a time $t$, and consider the execution of \textsf{GapFill} at the end of round $\tau = t$. We claim that the total extent of fractional play at time $t$ does not increase as we continue the execution of the algorithm from round $\tau=t$ to round $1$. To see why, let $C$ be a connected component at the end of round $\tau = t$ and let $h$ denote its head. If $\ptime^t(h) > t$ then no further {\bf advance} affects $C$ and hence it does not contribute to  an increase in the number of plays at time $t$. On the other hand, if $\ptime^t(h) \le t$, then even if $C$ is advanced in a subsequent round, each node $w$ of $C$ which ends up being played at $t$, i.e., has $\ptime^1(w) = t$ must have an ancestor $w'$ satisfying $\ptime^t(w') = t$, by the contiguity of $C$. Thus, \lref[Observation]{obs:treeflow} gives that $\sum_{u \in C : \ptime^1(u)=t} \prob(u) \le \sum_{u \in C : \ptime^t(u)=t} \prob(u)$. Applying this for each connected component $C$, proves the claim. Intuitively, any component which advances forward in time is only reducing its load/total fractional play at any fixed time $t$.
	
\begin{figure}[ht]
\centering
\subfigure[Connected components in the beginning of the algorithm]{
\includegraphics[scale=0.5]{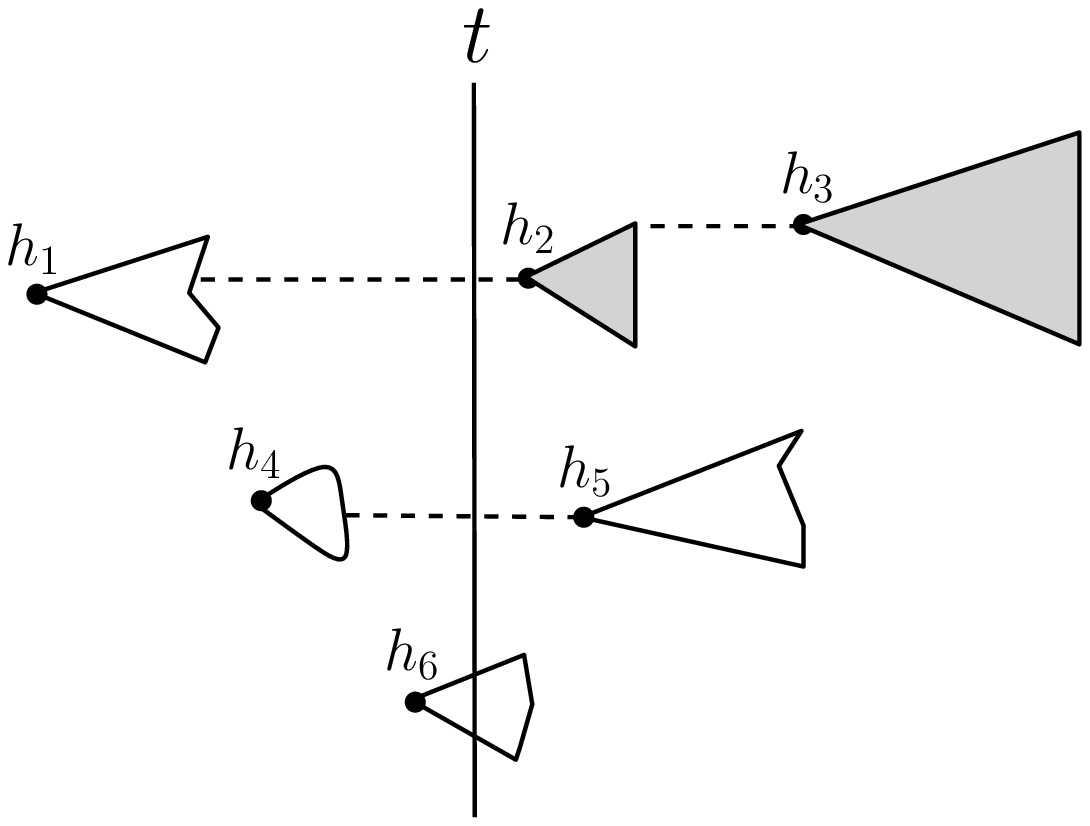}
\label{fig:subGapFill1}
}
\hspace{20pt} \subfigure[Configuration at the end of iteration $\tau = t$]{
\includegraphics[scale=0.5]{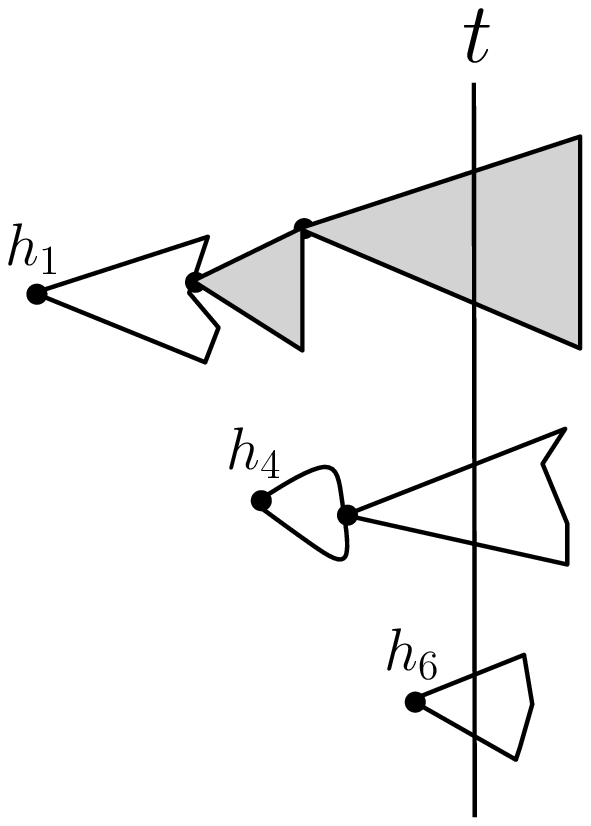}
\label{fig:subGapFill2}
}
\caption{Depiction of a strategy forest $\lpt(i,j)$ on a timeline, where each triangle is a connected component. In this example, $H = \{h_2, h_5\}$ and $C_{h_2}$ consists of the grey nodes. From Observation \ref{obs:treeflow} the number of plays at $t$ do not increase as components are moved to the left.}
\label{fig:gapFill}
\end{figure}
	
	Then consider the end of iteration $\tau = t$ and we now prove that the fractional extent of play at time $t$ is at most 3. Due to \lref[Lemma]{lem:convexppt}, it suffices to prove that $\sum_{u \in U} \prob(u) \le 2$, where $U$ is the set of nodes which caused an increase in the number of plays at time $t$, namely, $U = \{u : \ptime^{B + 1}(u) > t \textrm{ and } \ptime^t(u) = t\}$.
	
	Notice that a connected component of the original forest can only contribute to this increase if its head $h$ crossed time $t$, that is $\ptime^{B + 1}(h) > t$ and $\ptime^t(h) \le t$. However, it may be that this crossing was not directly caused by an {\bf advance} on $h$ (i.e. $h$ advanced till $\ptime^{B + 1}(\parent(h)) \ge t$), but an {\bf advance} to a head $h'$ in a subsequent round was responsible for $h$ crossing over $t$. But in this case $h$ must be part of the connected component of $h'$ when the latter {\bf advance} happens, and we can use $h'$'s advance to bound the congestion.
	
	To make this more formal, let $H$ be the set of heads of the original forest whose {\bf advances} made them cross time $t$, namely, $h \in H$ iff $\ptime^{B + 1}(h) > t$, $\ptime^t(h) \le t$ and $\ptime^{B + 1}(\parent(h)) < t$. Moreover, for $h \in H$ let $C_h$ denote the connected component of $h$ in the beginning of the iteration where an {\bf advance} was executed on $h$, that is, when $v$ was set to $h$ in \lref[Step]{alg:setV}. The above argument shows that these components $C_h$'s contain all the nodes in $U$, hence it suffices to see how they increase the congestion at time $t$.
	
In fact, it is sufficient to focus just on the heads in $H$. To see this, consider $h \in H$ and notice that no node in $U \cap C_h$ is an ancestor of another. Then \lref[Observation]{obs:treeflow} gives $\sum_{u \in U \cap C_h} \prob(u) \le \prob(h)$, and adding over all $h$ in $H$ gives $\sum_{u \in U} \prob(u) \le \sum_{h \in H} \prob(h)$.
	
	To conclude the proof, we upper bound the right hand side of the previous inequality. The idea now is that the play probabilities on the nodes in $H$ cannot be too large since their parents have $\ptime^{B + 1} < t$ (and each head has a large number of ancestors in $[1,t]$ because it was considered for an advance). More formally, fix $i,j$ and consider a head $h$ in $H \cap \lpt(i,j)$. From \lref[Step]{alg:gap1} of the algorithm, we obtain that $\depth(h) > (1/2) \ptime^{B + 1}(h) \ge t/2$. Since $\ptime^{B + 1}(\parent(h)) < t$, it follows that for every $d \le \lfloor t/2 \rfloor$, $h$ has an ancestor $u \in \lpt(i,j)$ with $\depth(u) = d$ and $\ptime^{B + 1}(u) \le t$. Moreover, the definition of $H$ implies that no head in $H \cap \lpt(i,j)$ can be an ancestor of another. Then again employing \lref[Observation]{obs:treeflow} we obtain
	\begin{align*}
		\sum_{h \in H \cap \lpt(i,j)} \prob(h) \le \sum_{u \in \lpt(i,j) : \depth(u) = d, \ptime^{B + 1}(u) \le t} \prob(u) \ \ \ \ \ \ \ (\forall d \le \lfloor t/2 \rfloor).
	\end{align*}
	Adding over all $i,j$ and $d \le \lfloor t/2 \rfloor$ leads to the bound $(t/2) \cdot \sum_{h \in H} \prob(h) \le \sum_{u : \ptime^{B + 1}(u) \le t} \prob(u)$. Finally, using \lref[Lemma]{lem:convexppt} we can upper bound the right hand side by $t$, which gives $\sum_{u \in U} \prob(u) \le \sum_{h \in H} \prob(u) \le 2$ as desired.
\end{proofof}

\subsection{Details of Phase~III (from Section~\ref{sec:phase-iii})}
\label{sec:details-phase-iii}


\proofof{Lemma~\ref{lem:visitprob}}
The proof is quite straightforward. Intuitively, it is because $\textsf{AlgMAB}$ (Algorithm~\ref{alg:roundmab}) simply follows the probabilities according to the transition tree $T_i$ (unless $\ptime(i,j,u) = \infty$ in which case it abandons the arm).
Consider an arm $i$ such that $\sigma(i) = j$, and any state $u \in \Si$. Let $\langle v_1 = \rho_i, v_2, \ldots, v_t = u \rangle$ denote the unique path in the transition tree for arm $i$ from $\rho_i$ to $u$.
Then, if $\ptime(i,j,u) \neq \infty$ the probability that state $u$ is played is exactly the probability of the transitions reaching $u$ (because in \lref[steps]{alg:mabPlay} and~\ref{alg:mabstep5}, the algorithm just keeps playing the states\footnote{We remark that while the plays just follow the transition probabilities, they may not be made contiguously.} and making the transitions, unless $\ptime(i,j,u) = \infty$). But this is precisely $\Pi_{k=1}^{t-1} p_{v_k, v_{k+1}} = \prob(i,j,u)/\prob(i,j,\rho_i)$ (from the properties of each strategy in the convex decomposition).
If $\ptime(i,j,u) = \infty$ however, then the algorithm terminates the arm in \lref[Step]{alg:mabAbandon} without playing $u$, and so the probability of playing $u$ is $0 = \prob(i,j,u)/\prob(i,j,\rho_i)$. This completes the proof.

%% file: app-dag.tex
\section{Proofs from Section~\ref{dsec:mab}}
\label{sec:app-dag}

\subsection{Layered DAGs capture all Graphs} 
\label{dsec:layered-enough}

We first show that \emph{layered DAGs} can capture all transition
graphs, with a blow-up of a factor of $B$ in the state space. For each
arm $i$, for each state $u$ in the transition graph $\S_i$, create $B$
copies of it indexed by $(v,t)$ for all $1 \leq t \leq B$. Then for each
$u$ and $v$ such that $p_{u,v} > 0$ and for each $1 \leq t < B$, place
an arc $(u,t) \to (v,t+1)$. Finally, delete all vertices that are not
reachable from the state $(\rho_i, 1)$ where $\rho_i$ is the starting
state of arm $i$. There is a clear correspondence between the
transitions in $\S_i$ and the ones in this layered graph: whenever state
$u$ is played at time $t$ and $\S_i$ transitions to state $v$, we have
the transition from $(u, t)$ to $(v, t + 1)$ in the layered
DAG. Henceforth, we shall assume that the layered graph created in this
manner is the transition graph for each arm.


%% file: mab-exploit.tex
\section{MABs with Budgeted Exploitation}
\label{xsec:mab}

As we remarked before, we now explain how to generalize the argument
from~\lref[Section]{sec:mab} to the presence of ``exploits''. A strategy
in this model needs to choose an arm in each time step and perform one
of two actions: either it \emph{pulls} the arm, which makes it
transition to another state (this corresponds to \emph{playing} in the
previous model), or \emph{exploits} it.  If an arm is in state $u$ and
is exploited, it fetches reward $r_u$, and cannot be pulled any more. As
in the previous case, there is a budget $B$ on the total number of pulls
that a strategy can make and an additional budget of $K$ on the total
number of exploits allowed. (We remark that the same analysis handles the case when pulling an arm also fetches reward, but for a clearer presentation we do not consider such rewards here.)

Our algorithm in~\lref[Section]{sec:mab} can be, for the large part, directly applied in this situation as well; we now explain the small changes that need to be done in the various steps, beginning with the new LP relaxation.
The  additional variable in the LP, denoted by $x_{u,t}$ (for $u \in \Si, t \in [B]$) corresponds to the probability of exploiting state $u$ at time $t$.
\begin{alignat}{2} \tag{$\mathsf{LP4}$} \label{xlp:mab}
	\max \ts \sum_{u,t} r_u &\cdot x_{u,t}\\
	w_{u,t} &= z_{\parent(u), t-1}  \cdot p_{\parent(u),u} & \qquad \forall t \in [2,B],\, u \in \S \label{xeq:mablp1}\\
	\ts \sum_{t' \le t} w_{u,t'} &\geq  \sum_{t' \leq t} (z_{u,t'} + x_{u,t'}) & \qquad \forall t \in [1,B], \, u \in \S \label{xeq:mablp2}\\
	\ts \sum_{u \in \S} z_{u,t} &\le 1 & \qquad \forall t \in [1,B]  \label{xeq:mablp3}\\
	\ts \sum_{u \in \S, t \in [B]}  x_{u,t} &\le K & \qquad \forall t \in [1,B]  \label{xeq:mablp4}\\
	w_{\rho_i, 1} &= 1 & \qquad \forall i \in [1,n] \label{xeq:mablp5}
\end{alignat}

\newcommand{\xsf}{\mathsf{x}\mathbb{T}}
\newcommand{\pull}{\mathsf{pull}}
\newcommand{\xpl}{\mathsf{exploit}}

\subsection{Changes to the Algorithm}

{\bf Phase I: Convex Decomposition}
\label{xsec:phase-i}

This is the step where most of the changes happen, to incorporate the
notion of exploitation.  For an arm $i$, its strategy forest $\xsf(i,j)$
(the ``\textsf{x}'' to emphasize the ``exploit'') is an assignment of
values $\ptime(i,j,u)$, $\pull(i,j,u)$ and $\xpl(i,j,u)$ to each state
$u \in \Si$ such that:
\begin{OneLiners}
\item[(i)] For $u \in \Si$ and $v = \parent(u)$, it holds that
  $\ptime(i,j,u) \geq 1+ \ptime(i,j,v)$, and
\item[(ii)] For $u \in \Si$ and $v = \parent(u)$ s.t $\ptime(i,j,u) \neq
  \infty$, then one of $\pull(i,j,u)$ or $\xpl(i,j,u)$ is equal to $p_{v,u}\,\pull(i,j,v)$ and the other is $0$; if
  $\ptime(i,j,u) = \infty$ then $\pull(i,j,u) = \xpl(i,j,u) = 0$.
\end{OneLiners}

For any state $u$, the value $\ptime(i,j,u)$ denotes
the time at which arm $i$ is \emph{played} (i.e., pulled or exploited) at state $u$, and $\pull(i,j,u)$ (resp. $\xpl(i,j,u)$) denotes the probability
that the state $u$ is pulled (resp. exploited). With the new definition, if
$\ptime(i,j,u) = \infty$ then this strategy does not play the arm at
$u$. If state $u$ satisfies $\xpl(i,j,u) \neq 0$, then strategy $\xsf(i,j)$ \emph{always exploits} $u$ upon reaching it and hence none of its descendants can be reached. For states $u$ which have $\ptime(i,j,u) \neq \infty$ and have $\xpl(i,j,u) = 0$, this strategy \emph{always pulls} $u$ upon reaching it. In essence, if $\ptime(i,j,u) \neq \infty$, either $\pull(i,j,u) = \pull(i,j,\rho_i) \cdot \pi_u$, or $\xpl(i,j,u) = \pull(i,j,\rho_i) \cdot \pi_u$.

Furthermore, these strategy forests are such that the following are also true.
\begin{OneLiners}
\item[(i)] $\sum_{j~\textsf{s.t}~\ptime(i,j,u)=t} \pull(i,j,u) = z_{u,t}$,
\item[(ii)] $\sum_{j~\textsf{s.t}~\ptime(i,j,u)=t} \xpl(i,j,u) = x_{u,t}$.
\end{OneLiners}
For convenience, let us define $\prob(i,j,u) = \pull(i,j,u) + \xpl(i,j,u)$, which denotes the probability of some play happening at $u$.

The algorithm to construct such a decomposition is very similar to the one presented in \lref[Section]{sec:details-phase-i}.
The only change is that in \lref[Step]{alg:convex2} of \lref[Algorithm]{alg:convex}, instead of looking at the first time when $z_{u,t} > 0$, we look at the first time when either $z_{u,t} > 0$ or $x_{u,t} > 0$. If $x_{u,t} > 0$, we ignore all of $u$'s descendants in the current forest we plan to peel off. Once we have such a collection, we again appropriately select the largest $\epsilon$ which preserves non-negativity of the $x$'s and $z$'s. Finally, we update the fractional solution to preserve feasibility. The same analysis can be used to prove the analogous of \lref[Lemma]{lem:convexstep} for this case, which in turn gives the desired properties for the strategy forests.

{\bf Phase II: Eliminating Small Gaps}
\label{xsec:phase-ii}

This is identical to the \lref[Section]{sec:phase-ii}.

{\bf Phase III: Scheduling the Arms}
\label{xsec:phase-iii}

The algorithm is also identical to that in \lref[Section]{sec:phase-iii}.
We sample a strategy forest $\xsf(i,j)$ for each arm $i$ and simply play connected components contiguously. Each time we finish playing a connected component, we
play the next component that begins earliest in the LP. The only
difference is that a play may now be either a \emph{pull} or an \emph{exploit} (which is deterministically determined once we fix a strategy forest); if this play is an exploit, the arm does not proceed to other states and is dropped. Again we let the algorithm run ignoring the pull and exploit budgets, but in the analysis we only collect reward from exploits which happen before either budget is exceeded.

The lower bound on the expected reward collected is again very similar to the previous model; the only change is to the statement of \lref[Lemma]{lem:beforetime}, which now becomes the following.

\begin{lemma} \label{xlem:beforetime}
  For arm $i$ and strategy $\xsf(i,j)$, suppose arm $i$ samples strategy
  $j$ in \lref[step]{alg:mabstep1} of \textsf{AlgMAB} (i.e., $\sigma(i)
  = j$). Given that the algorithm plays the arm $i$ in state $u$ during
  this run, the probability that this play happens before time
  $\ptime(i,j,u)$ \textbf{and} the number of exploits before this play is smaller than $K$, is at least $11/24$.
\end{lemma}

In \lref[Section]{sec:mab}, we showed \lref[Lemma]{lem:beforetime} by showing that
\[     \Pr [ {\tau}_u > \ptime(i,j,u) \mid \Evt_{iju} ] \leq
    \ts \frac{1}{2} \]
Additionally, suppose we can also show that
\begin{equation}
\Pr [ \textsf{number of exploits before }~u > (K-1) \mid \Evt_{iju} ] \leq
    \ts \frac{1}{24}  \label{eq:xpl-bound}
\end{equation}
Then we would have
\[     \Pr [( \textsf{number of exploits before }~u > (K-1) )\vee  ({\tau}_u > \ptime(i,j,u) ) \mid \Evt_{iju} ] \leq
    \ts 13/24, \]
which would imply the Lemma.

To show \lref[Equation]{eq:xpl-bound} we start with an analog of \lref[Lemma]{lem:visitprob} for bounding arm exploitations: conditioned on $\Evt_{i,j,u}$ and $\sigma(i') = j'$, the probability that arm $i'$ is exploited at state $u'$ before $u$ is exploited is at most $\xpl(i',j',u')/\prob(i',j',\rho_{i'})$. This holds even when $i' = i$: in this case the probability of arm $i$ being exploited before reaching $u$ is zero, since an arm is abandoned after its first exploit. Since $\sigma(i') = j'$ with probability $\prob(i',j',\rho_{i'})/24$, it follows that the probability of exploiting arm $i'$ in state $u'$ conditioned on $\Evt_{i,j,u}$ is at most $\sum_{j'} \xpl(i',j',u')/24$. By linearity of expectation, the expected number of exploits before $u$ conditioned on $\Evt_{i,j,u}$ is at most $\sum_{(i',j',u')} \xpl(i',j',u')/24 = \sum_{u',t} x_{u,t}/24$, which is upper bounded by $K/24$ due to LP feasibility. Then \lref[Equation]{eq:xpl-bound} follows from Markov inequality.

The rest of the argument is identical to that in \lref[Section]{sec:mab} giving us the following.
\begin{theorem}
There is a randomized $O(1)$-approximation algorithm for the \mab problem with an exploration budget of $B$ and an exploitation budget of $K$.
\end{theorem}


%% file: main.bbl
\begin{thebibliography}{GGM06}

\bibitem[Ber05]{Bert05}
Dimitri~P. Bertsekas.
\newblock {\em Dynamic programming and optimal control.}
\newblock Athena Scientific, Belmont, MA, third edition, 2005.

\bibitem[BGK11]{BGK11}
Anand Bhalgat, Ashish Goel, and Sanjeev Khanna.
\newblock Improved approximation results for stochastic knapsack problems.
\newblock In {\em SODA '11}. Society for Industrial and Applied Mathematics,
  2011.

\bibitem[BL97]{BirgeL97}
John~R. Birge and Fran{\c{c}}ois Louveaux.
\newblock {\em Introduction to stochastic programming}.
\newblock Springer Series in Operations Research. Springer-Verlag, New York,
  1997.

\bibitem[CR06]{ChawlaR06}
Shuchi Chawla and Tim Roughgarden.
\newblock Single-source stochastic routing.
\newblock In {\em Proceedings of APPROX}, pages 82--94. 2006.

\bibitem[Dea05]{Dean-thesis}
Brian~C. Dean.
\newblock {\em Approximation Algorithms for Stochastic Scheduling Problems}.
\newblock PhD thesis, MIT, 2005.

\bibitem[DGV05]{dgv05}
Brian~C. Dean, Michel~X. Goemans, and Jan Vondr{\'a}k.
\newblock Adaptivity and approximation for stochastic packing problems.
\newblock In {\em SODA}, pages 395--404, 2005.

\bibitem[DGV08]{DeanGV08}
Brian~C. Dean, Michel~X. Goemans, and Jan Vondr{\'a}k.
\newblock Approximating the stochastic knapsack problem: The benefit of
  adaptivity.
\newblock {\em Math. Oper. Res.}, 33(4):945--964, 2008.

\bibitem[GGM06]{GGM06}
Ashish Goel, Sudipto Guha, and Kamesh Munagala.
\newblock Asking the right questions: model-driven optimization using probes.
\newblock In {\em PODS}, pages 203--212, 2006.

\bibitem[GI99]{GI99}
Ashish Goel and Piotr Indyk.
\newblock Stochastic load balancing and related problems.
\newblock In {\em 40th Annual Symposium on Foundations of Computer Science (New
  York, 1999)}, pages 579--586. IEEE Computer Soc., Los Alamitos, CA, 1999.

\bibitem[Git89]{Gittins89}
J.~C. Gittins.
\newblock {\em Multi-armed bandit allocation indices}.
\newblock Wiley-Interscience Series in Systems and Optimization. John Wiley \&
  Sons Ltd., Chichester, 1989.
\newblock With a foreword by Peter Whittle.

\bibitem[GKN09]{GoelKN09}
Ashish Goel, Sanjeev Khanna, and Brad Null.
\newblock The ratio index for budgeted learning, with applications.
\newblock In {\em SODA '09: Proceedings of the twentieth Annual ACM-SIAM
  Symposium on Discrete Algorithms}, pages 18--27, Philadelphia, PA, USA, 2009.
  Society for Industrial and Applied Mathematics.

\bibitem[GM07a]{GuhaM-stoc07}
Sudipto Guha and Kamesh Munagala.
\newblock Approximation algorithms for budgeted learning problems.
\newblock In {\em S{TOC}'07---{P}roceedings of the 39th {A}nnual {ACM}
  {S}ymposium on {T}heory of {C}omputing}, pages 104--113. ACM, New York, 2007.
\newblock Full version as \emph{Sequential Design of Experiments via Linear
  Programming}, \url{http://arxiv.org/abs/0805.2630v1}.

\bibitem[GM07b]{GuhaM-soda07}
Sudipto Guha and Kamesh Munagala.
\newblock Model-driven optimization using adaptive probes.
\newblock In {\em SODA '07: Proceedings of the eighteenth annual ACM-SIAM
  symposium on Discrete algorithms}, pages 308--317, Philadelphia, PA, USA,
  2007. Society for Industrial and Applied Mathematics.
\newblock Full version as \emph{Adaptive Uncertainty Resolution in Bayesian
  Combinatorial Optimization Problems}, \url{http://arxiv.org/abs/0812.1012v1}.

\bibitem[GM09]{GuhaM09}
Sudipto Guha and Kamesh Munagala.
\newblock Multi-armed bandits with metric switching costs.
\newblock In {\em ICALP}, pages 496--507, 2009.

\bibitem[GMP11]{GuhaMP11}
Sudipto Guha, Kamesh Munagala, and Martin Pal.
\newblock Iterated allocations with delayed feedback.
\newblock {\em ArXiv}, arxiv:abs/1011.1161, 2011.

\bibitem[GMS07]{GuhaMS07}
Sudipto Guha, Kamesh Munagala, and Peng Shi.
\newblock On index policies for restless bandit problems.
\newblock {\em CoRR}, abs/0711.3861, 2007.
\newblock \url{http://arxiv.org/abs/0711.3861}. Full version of
  \emph{Approximation algorithms for partial-information based stochastic
  control with Markovian rewards} (FOCS'07), and \emph{Approximation algorithms
  for restless bandit problems}, (SODA'09).

\bibitem[KRT00]{KRT-sched}
Jon Kleinberg, Yuval Rabani, and {\'E}va Tardos.
\newblock Allocating bandwidth for bursty connections.
\newblock {\em SIAM J. Comput.}, 30(1):191--217 (electronic), 2000.

\bibitem[MSU99]{MohringSU99}
Rolf~H. M\"{o}hring, Andreas~S. Schulz, and Marc Uetz.
\newblock Approximation in stochastic scheduling: the power of lp-based
  priority policies.
\newblock {\em Journal of the ACM (JACM)}, 46(6):924--942, 1999.

\bibitem[Pin95]{Pinedo}
Michael Pinedo.
\newblock {\em Scheduling: Theory, Algorithms, and Systems}.
\newblock Prentice Hall, 1995.

\bibitem[SU01]{SkutU01}
Martin Skutella and Marc Uetz.
\newblock Scheduling precedence-constrained jobs with stochastic processing
  times on parallel machines.
\newblock In {\em Proceedings of the twelfth annual ACM-SIAM symposium on
  Discrete algorithms}, pages 589--590. Society for Industrial and Applied
  Mathematics, 2001.

\end{thebibliography}
